\documentclass[12pt]{article}
\usepackage[margin=1.25in]{geometry}
\usepackage{amsfonts,amsmath,graphicx}
\usepackage{verbatim,xcolor}
\usepackage{cite}
\usepackage{hyperref}
\usepackage{amsthm}
\usepackage{color,soul}

\usepackage{graphicx}
\usepackage{amsfonts,amssymb,amsmath}
\usepackage{mathptmx}
\usepackage{latexsym}
\usepackage{verbatim,xcolor}
\usepackage{xcolor}
\usepackage{bm}

\newcommand{\FF}{\mathbb{F}}

\newcommand{\dv}{{\rm d}}

\newtheorem{theorem}{Theorem}
\newtheorem{lemma}{Lemma}

\newtheorem{proposition}{Proposition}
\newtheorem{definition}{Definition}

\newtheorem{formula}{Formula}

\def\hs4{ \ \ \ \ }
\begin{document}

\title{Excitation spectrum for the Lee-Huang-Yang model for the bose gas in a periodic box}

\date{\today}

\author{Stephen Sorokanich\thanks{The work of the author was partially supported by Grant No.~1517162 of the Division of Mathematical Sciences (DMS) of the National Science Foundation (NSF).}
}

\maketitle

\begin{abstract}
   We derive a new formula for excited many-body eigenstates of the approximate Hamiltonian of Lee, Huang, and Yang \color{black}\cite{Lee Huang Yang I, Lee Yang I, Lee Yang II}\color{black},~ which describes the weakly interacting bose gas in a periodic box. Our formula utilizes a non-Hermitian transform to isolate a process of \textit{pair excitation} in all excited states. This program is inspired by the description of Wu \color{black}\cite{Wu61,Wu98}\color{black}~ for the many-body ground state of the Lee-Huang-Yang Hamiltonian, as well as similar descriptions for the ground state of quadratic Hamiltonians in generic trapping potentials \color{black}\cite{Fetter}\color{black}.  
\end{abstract}
\section{Introduction}
\label{sec:Intro}

\color{black} Quadratic Bosonic Hamiltonians have been the subject of increasing mathematical attention. Rigorous results now exist which describe the approximation of many-body Hamiltonians for the interacting bose gas by effective quadratic Hamiltonians, as well as the mathematical foundations for their diagonalization via unitary rotation. Both of these areas of study display mathematical methods which were originally inspired by techniques developed in physics.\cite{Napiorkowski, Seiringer, Derezinski}\color{black}~

A less popular method in the physical study of the Bose gas involves the non-unitary transformation of approximate Hamiltonians \color{black}\cite{Wu61,Fetter}\color{black}. Such transformations describe a pairwise-scattering process (\textit{pair excitation}) which results in the partial depletion of the condensate to excited states, and is hypothesized to contribute to corrections to the ground-state wavefunction beyond mean-field descriptions \color{black}\cite{Machedon}\color{black}. The use of such a method has been present since early on in the study of the Bose gas. In 1961, Wu \color{black}\cite{Wu61}\color{black}~ applied a bounded exponential transformation to a particle-conserving approximate Hamiltonian, $\mathcal{H}_\mathrm{Wu}$, for repulsively-interacting bosons in a trapping potential. A kernel $k(x,y)$ describes the spatial variation of this scattering process via the second quantized operator

\begin{equation}\label{Wuexp}
\mathcal{W}:=\frac{(a_{\overline{\phi}})^2}{2N}\int{dxdy\{k(x,y)a^\ast_{\perp,x}a^\ast_{\perp,y}\}},
\end{equation}
which is introduced in order to transform the Hamiltonian via $\exp(\mathcal W)\mathcal{H}_\mathrm{Wu}\exp(-\mathcal W)$. The kernel is specified \textit{a posteriori} so that the transformed Hamiltonian contains no terms proportional to the product of field operators $(a^\ast_{\perp,x} a^\ast_{\perp,y})$. We recently placed this method on a rigorous footing by proving the existence of solutions to the integro-differential equation for $k(x,y)$ under mild assumptions on the interaction potential \color{black}\cite{Grillakis Margetis Sorokanich}\color{black}. The structure of the transformed Hamiltonian $\exp(\mathcal W)\mathcal H_\mathrm{Wu}\exp(-\mathcal W)$ allowed us to write a simple formula for all excited many-body states by making a specific choice of basis. 

In the present work, we show the utility of a similar non-Hermitian approach in the context of a quadratic and non-particle-conserving approximate Hamiltonian based on the famous model of Lee, Huang and Yang \color{black}\cite{Lee Huang Yang I, Lee Yang I, Lee Yang II} \color{black}. This model describes low-lying excitations of the hard-sphere bose gas in a periodic box with volume $L^3$ via the following second-quantized Hamiltonian in the momentum basis:
\begin{equation}\label{HLHY}
\mathcal{H}_\mathrm{LHY} = 4\pi a \rho N + \sum_{k\in\mathbb{Z}^3_L,\, k\not= 0}{\big\{k^2+8\pi a \rho\big\}a^\ast_k a_k} + 4\pi a \rho\sum_{k\in\mathbb{Z}^3_L,\, k\not= 0}{\big\{a^\ast_k a^\ast_{-k}+a_k a_{-k}\big\}}.
\end{equation}
In this setting we introduce an analog to the operator \eqref{Wuexp}:
\begin{equation}\label{LHYexp}
\mathcal{P}:=\sum_{k\in\mathbb{Z}^+_L}{-\alpha(k) a_k^\ast a^\ast_{-k}},
\end{equation}
where the real parameter $\alpha(k)$ is the Fourier transform of a translation-invariant kernel $k(x,y)$.
We emphasize that the transformed Hamiltonian $\exp(-\mathcal{P})\mathcal{H}_\mathrm{LHY}\exp(\mathcal{P})$ is now both non-Hermitian and \emph{unbounded}, and the point spectrum $\sigma_\mathrm{p}(\mathcal{H}_\mathrm{LHY})$ is not equal to the point spectrum $\sigma_\mathrm{p}(\exp(-\mathcal{P})\mathcal{H}_\mathrm{LHY}\exp(\mathcal{P}))$. In what follows, we show that the spectrum of $\mathcal{H}_\mathrm{LHY}$ can be `recovered' in a straightforward manner from the spectrum of this transformed operator, and surprisingly, all excited state wavefunctions for $\mathcal{H}_\mathrm{LHY}$ have a particularly simple expression by utilizing the operator $\exp(\mathcal{P})$. Wu observed that the ground state for the Lee-Huang-Yang Hamiltonian can be written using this $\mathcal{P}$:
$$
|\Psi_0\rangle = \exp(-\mathcal{P})|vac\rangle.
$$
We extend this result to show that all excited states of $\mathcal{H}_\mathrm{LHY}$ take the form 
\begin{equation}\label{essential eigenstate}
|\Psi_{E(\vec{n})}\rangle :=\exp(-\mathcal{P})|\Psi_{\vec{n}}\rangle,
\end{equation}
where $|\Psi_{\vec{n}}\rangle$ is a \emph{finite superposition} of momentum tensor product states (see Formula \eqref{formula1}). This is a unique feature of using the non-Hermitian pair excitation transform---factoring out the operator $\exp(\mathcal{P})$ yields a simple formula for excited many-body states in a physically transparent basis.

\subsubsection*{A note on unitary rotations of quadratic Hamiltonians}


It is true that all information about the spectrum and excited states of Hamiltonian \eqref{HLHY} can be derived by means of its unitary transformation. The operators defined by
$$
b_k := \cosh(\alpha(k))a_k + \sinh(\alpha(k))a_{-k}^*,\quad\mathrm{for}\quad 0\le\alpha(k)<1 
$$
satisfy the canonical commutation relations of bosonic operators, and 
the particular choice of $\alpha(k)$ that we will use in \eqref{LHYexp} makes $\mathcal{H}_\mathrm{LHY}$ diagonal in  $\{b^\ast_k, b_k\}$,  that is
\begin{equation}
\mathcal{H}_\mathrm{LHY}= \sum_{k}{\epsilon(k) b_k^\ast b_k},\quad 0<\epsilon(k)<\epsilon({k'})\quad\forall\quad |k|<|k'|.
\end{equation}

This transformation can alternatively be expressed as $b_k = \exp(\mathcal{X})a_k \exp(-\mathcal{X})$, where $\mathcal{X}$ is the Hermitian operator
\begin{equation}\label{unitary transformation}
\mathcal{X} := \sum_{k\in\mathbb{Z}^+_L}{\alpha(k)(a^\ast_k a^\ast_{-k} - a_k a_{-k})}.
\end{equation}
Excited many-body states of $\mathcal{H}_\mathrm{LHY}$ are given by tensor products of the $b$ operators
\begin{equation}\label{quasiparticle}
|\Psi_{\{n_k\}}\rangle:=\prod_{k\in\mathbb{Z}^+_L}{(b_k^\ast)^{n_k}}|\Psi_0\rangle,\quad n_k\in\mathbb{N}\quad \forall k,\quad \sum_{k}{n_k}<\infty.
\end{equation}
The ground state in particular exhibits form 
$$
|\Psi_0\rangle =\exp(\mathcal{X})|vac\rangle = \exp(-\mathcal{P})|vac\rangle.
$$

The fact that $\mathcal{H}_\mathrm{LHY}$ is diagonal in the operators $\{b_k,b^\ast_k\}$ leads one to attribute physical significance to them as \textit{quasiparticles}. As far as the computation of observables goes, one may use the operator basis $\{b_k,b_k^\ast\}$ just as easily as the momentum operators $\{a_k,a_k^\ast\}$.

But expressions in terms of the $b$ operators do not reveal the pair-excitation structure of \textit{all} excited states. Compare this to the schematic \eqref{essential eigenstate} which holds for all many-body excited states.

Additionally, starting from a state expressed by \eqref{quasiparticle} and attempting to expand this state in the momentum basis becomes unwieldy. The presence of both the annihilation operator $a_k$ as well as the creation operator $a_{-k}^\ast$ in the definition of $b_k^\ast$ means that, for example, the state $$|\Psi_{\{n_k,n_{-k}\}}\rangle:=(b_k^\ast)^{n_k} (b^\ast_{-k})^{n_{-k}}|\Psi_0\rangle,\quad n_k,n_{-k}\in\mathbb{N}$$ 
will require many commutations in order to derive a formula containing only products of the $a^\ast_k$ operators. For a state with $n_k=n=n_{- k}$, we will have to perform at least $2^n$ commutations in order to find such an expression. Our method, in contrast, provides a formula for excited states of this system explicitly in terms of momentum pairs by accounting for the non-perturbative effect of pair excitation via the operator $\exp(\mathcal{P})$.

In pursuing a rigorous analysis of the Lee-Huang-Yang model, we do not address the validity of quadratic models for BEC in this work. See \color{black}\cite{Seiringer}~\color{black} for a detailed discussion of these matters in the periodic setting. As such, we rely on the ubiquity of these approximations in many areas of physics for our motivation.


We choose the particular model of Lee, Huang and Yang as the subject of the current work because in our opinion it captures the essential physics of pair excitation in the most concise fashion. In this vein, the Lee-Huang-Yang Hamiltonian describes low-energy two-particle collisions via a delta-function interaction potential and a scalar \textit{scattering length}. We note that this introduces a familiar blow-up to the energy spectrum, which will be accounted for in a standard way. Our construction in fact holds for a variety of quadratic Hamiltonians; for example, Bogoliubov adopted the following quadratic Hamiltonian for bosons in the periodic box with a two-body interaction potential $v$ \color{black}\cite{Bogoliubov, Lieb}\color{black}
\begin{equation}\label{Bog}
\mathcal{H}_\mathrm{Bog} := \sum_{k}{(k^2+\frac{N}{2V}\hat{v}(k))a^\ast_k a_k}+\frac{N}{2V}\sum_{k}{\hat{v}(k)(a^\ast_k a^\ast_{-k} + a_k a_{-k}}).
\end{equation}

\subsubsection*{Outline of Paper}
In section \eqref{sec:LY-H}, we introduce the many-body Hamiltonian and describe the heuristic approximation scheme by which the Lee-Huang-Yang Hamiltonian, $\mathcal{H}_\mathrm{LHY}$, is derived. 

In section \eqref{sec:construction} we introduce the pair excitation transformation, $\exp(\mathcal{P})$, and provide an exact formula for excited states of the non-unitary transformed Hamiltonian via expansion in the momentum basis. 

In sections \eqref{sec:non-uni-family} and \eqref{completeness} we discuss the spectral theory for the transformed operator $$\exp(-\mathcal{P})\mathcal{H}_\mathrm{LHY}\exp(\mathcal{P}),$$ whose spectrum is \emph{not} identical to the spectrum of $\mathcal{H}_\mathrm{LHY}$. This fact necessitates that we take care in identifying the mappings of LHY eigenstates among the eigenstates of the transformed problem; they are exactly those eigenstates of the transformed system which are finite superpositions of momentum states.

In section \eqref{PC Ham}, the formula for eigenstates is adapted to a related particle-conserving Hamiltonian that appears in the course of the approximation scheme for $\mathcal{H}_\mathrm{LHY}$. We associate this approximation with Wu, as it is directly implied by his work \color{black}\cite{Wu61,Wu98}\color{black}.

We conclude by discussing possible extensions of this model to the non-translation invariant setting.

\subsection*{Notation and preliminary results}
\label{subsec:notation}

Our domain will be the periodic box in 3 dimensions with length $L$, which we denote $B_L=[0,L]^3$, and its volume $|B_L|=L^3$. Unless otherwise noted, all integrals are assumed to be over $B_L$.

Function spaces on $\mathbb{R}^{3}$ or $B_L$ are denoted by lowercase gothic letters, viz., 
\begin{equation*}
    \mathfrak{h}(\mathbb{R}^3) := L^2(\mathbb{R}^3)~.
\end{equation*}
An exception is the definition  $\phi^{\perp}:=
\big\{e\in \mathfrak{h}\ \big\vert\ e\perp\phi\big\}$ where $\phi\in \mathfrak{h}$ is the condensate wave function. When considering function spaces for many-particle systems, we will take 
$$\mathfrak{h}(\mathbb{R}^{3N}):=L_\mathrm{sym}^2(\mathbb{R}^{3N}),$$ where the subscript refers to the fact that the functions in this space are symmetric under permutations of the $N$ coordinates, $x_1,x_2,\dots,x_N\in\mathbb{R}^3$. Operators on $\mathfrak{h}$ are given by greek or roman letters. For example, $\delta(x,y)$ will denote the Dirac mass of the identity operator and $\Delta$ is the Laplace operator. The many-body Hamiltonian on the configuration space $\mathfrak{h}(\mathbb{R}^{3N})$ is denoted $H_N$ which we define in section \eqref{sec:LY-H}.

 
For $f, g \in \mathfrak{h}$
the tensor-product corresponding to $f(x)\overline{g(y)}$ is expressed by $f\otimes g$. The symmetrized tensor product of $f,g$ is 
\begin{equation*}
f\otimes_{\mathrm{s}}g:= \frac{1}{\sqrt{2}}\big\{f\otimes g +g\otimes f\big\}~.
\end{equation*}
\subsubsection*{The bosonic Fock space}
We define the Bosonic Fock space, $\mathbb{F}$, as a direct sum of $n$-particle Hilbert spaces, via
\begin{equation*}
\FF:=\bigoplus_{n=0}^{\infty}\mathbb{F}_{n}~;\quad \FF_0:=\mathbb{C}~,\quad \FF_n:=\mathfrak{h}(\mathbb{R}^{3n})\quad \mbox{if}\ n\ge 1~. 
\end{equation*}
Vectors in $\FF$ are described as sequences of $n$-particle wave functions, or using ket notation, as in  $\vert u\rangle=\{u^n\}$ where $u^n\in \mathfrak{h}(\mathbb{R}^{3n})$, for $n\ge0$. The inner product of $\vert u\rangle=\{u^n\},\, \vert w\rangle=\{w^n\}\in\FF$ is 
\begin{equation*}
\langle u,w\rangle_{\FF}:=\sum_{n=0}^\infty{\langle \overline{u}^n,w^n\rangle_{L^2(\mathbb{R}^{3n})}}~,
\end{equation*}
This induces the norm $\Vert |u\rangle\Vert_\mathbb{F}=\sqrt{\langle u,u\rangle_{\FF}}$.  
As we already aluded to, we employ the bra-ket notation for Schr\"odinger state vectors in $\FF$ to distinguish them from wave functions in $\mathfrak{h}(\mathbb{R}^{3n})$. 

Operators on $\mathbb{F}$ will be denoted by calligraphic letter, an example being the Hamiltonian $\mathcal{H}:\mathbb{F}\to\mathbb{F}$. This, of course, excludes the annihilation and creation operators, including the field operators $a_x,\,a^\ast_x,$ as well as the operators $a_{\overline \phi}, a^*_\phi$  and $a_k, a_k^{\ast}$ associated with the basis $\{e_k(x)\}$ discussed shortly. The vacuum state in $\FF$ is $\vert  vac\rangle:={\{1,0,0,\dots\}}$, where the unity is placed in the zeroth slot. A symmetric $N$-particle wave function, $\psi_N\in \mathfrak{h}(\mathbb{R}^{3N})$, has a natural embedding into $\FF$ given by $\vert\psi\rangle_N=\{0,0,\dots,\psi_N(x),0,\dots\}$, where $\psi_N(x)$ is in the $N$-th slot. The set of all vectors $\vert \psi\rangle_N$ for $N$ fixed is a linear subspace of $\mathbb{F}$, denoted $\mathbb{F}_N$, and is called the `$N$-th fiber' ($N$-particle sector) of $\FF$. We sometimes omit the subscript `$N$' when referring to a $\vert\psi\rangle_N\in\mathbb{F}_N$ when the context makes it clear . 

A Hamiltonian on $\mathfrak{h}(\mathbb{R}^{3N})$ admits an extension to an operator on $\FF$. This extension is carried out via the Bosonic field operator $a_x$ and its adjoint, $a_x^\ast$, for spatial coordinate $x\in\mathbb{R}^3$. To define these field operators, first consider the annihilation and creation operators for a one-particle state $f\in \mathfrak{h}$, denoted by $a(\overline{f})$ and $a^\ast(f)$. These operators act on $\vert u\rangle=\{u^n\}\in\FF$ according to 
\begin{equation*}
\big(a(\overline{f})\vert u\rangle\big)^n:=\sqrt{n+1}\int{\dv x\, \overline{f(x)}\,u^{n+1}(x,x_2,\dots,x_n)}~,
\end{equation*}
\begin{equation*}
\big(a^\ast(f)\vert u\rangle\big)^n:=\frac{1}{\sqrt{n}}\sum_{j\le n}{f(x_j)\,u^{n-1}(x_1,\dots,x_{j-1},x_{j+1},\dots,x_n)}~.
\end{equation*}
We often use the symbols $a_{\overline f}:=a(\overline f)$ and $a^*_f:= a^*(f)$. Also, given an orthonormal basis, $\{e_j(x)\}_j\subset \mathfrak{h}$, we will write $a_j^\ast$ in place of $a^\ast (e_j)$ and $a_j$ in place of $a(\overline{e_j})$.  

The Boson field operators $a_x^\ast,\, a_x$ are now defined using an orthonormal basis via 
\begin{equation*}
a^\ast_x =\sum_j{e_j(x)\,a^\ast_j}~,\quad 
a_x = \sum_j{\overline{e_j(x)}\,a_j}~.
\end{equation*}
The canonical commutation relations 
$[a_x,a_y^\ast] =\delta(x-y)$, $[a_x,a_y] = 0$ then follow.

An orthonormal basis that we use extensively in this work consists of the  momentum eigenfunctions on $B_L$:
\begin{equation}
e_k(x) := \frac{e^{ik\cdot x}}{\sqrt{|B_L|}},\quad\mathrm{where}\quad k:= \frac{2\pi n}{L},\quad n\in\mathbb{Z}^3.
\end{equation}
Thus we consider periodic functions (of spatial variable $x\in B_L$) with period $L$ and denote the dual lattice $$\mathbb{Z}^3_L:= \{k=2\pi n/L \,\big|\,n\in\mathbb{Z}^3\}.$$ By virtue of the commutation relations on $\{a_x, a_x^\ast\}$, the creation and annihilation operators for the states $\{e_k(x)\}_{k\in\mathbb{Z}^3_L}$, denoted $a(\overline{e_k}):=a_k$ and $a^\ast(e_k):=a^\ast_k$ 
satisfy
$$[a_{k_1},a^\ast_{k_2}]=\delta(k_1,k_2),\quad [a_{k_1},a_{k_2}]=0=[a^\ast_{k_1},a^\ast_{k_2}],\quad k_1,k_2\in\mathbb{Z}^3_L.$$

Vectors  $|\Psi\rangle\in\mathbb{F}$ can also be expressed in terms of the occupation number basis for $\{e_k(x)\}$. The orthonormal elements of this Fock space basis consist of tensor product states for every collection of integers $\{n_k\in\mathbb{N}\}_{k\in\mathbb{Z}^3_L}$ with $\sum_{k}{n_k}<\infty$, given by 
$$|n_{k_1},n_{k_2},\dots\rangle := \prod_{k\in\mathbb{Z}^+_3}\frac{\big(a(e_{k})^\ast\big)^{n_{k}}}{\sqrt{n_{k}!}}|vac\rangle.$$
Fixing a finite collection of momenta, $(k_1,\dots,k_n)$, we will also consider the states $|n_{k_1},n_{k_2},\dots,n_{k_n}\rangle$ which are defined in a similar manner.

Finally, the symbol ``$\approx$" will be used in two senses. The first sense refers the the heuristic approximation of Fock space operators, as in $\mathcal{H}_\mathrm{LHY}\approx \mathcal{H}$. The second sense is the precise relation of asymptotic equivalence of functions. A function $f(z)$ is said to be asymptotically equivalent to $g(z)$ as $|z|\to\infty$, i.e., $f(z)\approx g(z)$ as $|z|\to\infty$, provided $\lim_{|z|\to\infty}\frac{f(z)}{g(z)}=1.$ We will be careful to clarify in which sense we are using this symbol when it occurs in the text.

\textbf{On the operator exponential:} We will make extensive use of the operator $\exp(\mathcal{P}):\mathbb{F}\to \mathbb{F}$, where
$$\mathcal{P} := \sum_{k\in\mathbb{Z}^+_L}{-\alpha(k) a_k^*a_{-k}^*},\quad\mathrm{for}\quad 0\le \alpha(k)<1\quad\forall k\in\mathbb{Z}^+_L.$$
The operator $\mathcal{P}$ is unbounded on $\mathbb{F}$. While it does not have any eigenvalues, its spectrum $\sigma(\mathcal{P})$ (that is, the union of point, continuous and residual spectra) consists of the entire complex plane. A precise definition of $\mathrm{dom}(\exp(\mathcal{P}))$ is not strictly necessary for our purposes; it will suffice to consider this operator as a formal series in powers of $\mathcal{P}$ as long as we are acting on states which remain finite-norm under the action of this formal operator series. Fixing $k\in\mathbb{Z}^+_L$, a state vector of the form
$$|\Psi(k)\rangle = \sum_{n=0}^\infty{c_n |n,n\rangle},\quad|n,n\rangle := \frac{(a_k^\ast a_{-k}^\ast)^n}{n!}|vac\rangle, \quad\sum_{n=0}^\infty{|c_n|^2}=1$$
will belong to $\mathrm{dom}\big(\exp(\mathcal{P})\big)$ provided that the sequence 
\begin{equation}
\Big\{\tilde{c}_s := \sum_{n=0}^{s}{c_n(-\alpha(k))^{s-n}\binom{s}{s-n}}\Big\}_{s=0}^\infty
\end{equation}
is square-summable. We note that the use of exponentials involving creation operators such as $\exp(\mathcal{P})$ is quite common the study of generalized coherent states \color{black}\cite{Perelomov, Klauder}\color{black}.

The following lemma will be crucial.
\begin{lemma}\label{lemma1}
Any two operators $\mathcal{A},\mathcal{B}$ in the Fock space satisfy the identity
$$
e^\mathcal{A}\mathcal{B}e^{-\mathcal{A}} = e^{ad(\mathcal{A})}\mathcal{B} = \sum_{n=0}^\infty{\frac{ad(\mathcal{A})^n(\mathcal{B})}{n!}},
$$
where $ad^0(\mathcal{A})(\mathcal{B}) = \mathcal{B}$ and 
$$
ad(\mathcal{A})^n(\mathcal{B}) = [\mathcal{A},ad(\mathcal{A})^{n-1}(\mathcal{B})],\quad n\ge 1.\quad\Box
$$
\end{lemma}

\section{Many-body Hamiltonian and approximation scheme}\label{sec:LY-H}

We now summarize the bosonic many-body Hamiltonian and discuss the quadratic approximation of Lee, Huang, and Yang. This is included as motivation for the analysis that follows; we do not rigorously justify the derivation of this section.

Consider $N$ identical bosons inside the box $B_L$, with periodic boundary conditions and repulsive pairwise particle interaction $\upsilon$. On the Hilbert space $\mathfrak{h}(\mathbb{R}^{3N})$ of symmetric $N-$particle wavefunctions, the Hamiltonian for this system reads 
\begin{equation}
H_N = \sum_{j=1}^N{-\Delta_j}+\sum_{i<j}^N{\upsilon(x_i,x_j)},\quad x_j\in\mathbb{R}^3.
\end{equation}
Here we choose units such that $\hbar=2m=1$, where $\hbar$ is Planck's constant, and $m$ is the atomic mass. The interaction potential $\upsilon$ should be understood to be positive, symmetric, and compactly supported. 

This Hamiltonian can be lifted to the bosonic Fock space via the field operators $\{a_x, a_x^\ast\}$:
\begin{equation}\label{FockHamiltonian}
\mathcal{H} = \int{dx\big\{a_x^\ast (-\Delta_x) a_x\big\}}+ \frac{1}{2}\iint{dxdy\big\{\upsilon(x,y)a_x^\ast a_y^\ast a_x a_y\big\}}.
\end{equation}

In the spirit of Lee, Huang and Yang, we take the interaction $\upsilon$ to be the Fermi pseudopotential, which is an effective operator that reproduces the low-energy limit of the far field of the exact wavefunction \color{black}\cite{Huang Yang Luttinger}\color{black}.  If $f$ is any two-body wavefunction, we therefore take 
\begin{equation}\label{pseudopotential}
\upsilon(x_i,x_j)f(x_i,x_j)\approx g\delta(x_i-x_j)\frac{\partial}{\partial x_{ij}}\big[x_{ij}f(x_i,x_j)\big],\quad (i\not=j),
\end{equation}
where $x_{ij}:=|x_i-x_j|$, and $g := 8\pi a$ where $a$ is the scattering length. We further simplify this interaction by omitting $(\partial/\partial x_{ij})x_{ij}$ from \eqref{pseudopotential}, viz., 
\begin{equation}\label{delta}
\upsilon(x_i,x_j)\mapsto g \delta(x_i-x_j).
\end{equation}
The substitution \eqref{delta} will be exact for solutions to the many-body wavefunction with sufficient regularity \color{black}\cite{Huang Yang Luttinger}\color{black}. Inserting \eqref{delta} into \eqref{FockHamiltonian}, and expanding in the momentum basis yields
\begin{equation}\label{Ham1}
 \mathcal{H} \approx \sum_{k\in\mathbb{Z}^3_L}{|k|^2 a^\ast_k a_k}+ \frac{4\pi a}{|B_L|} \sum_{k_1+k_2=k_3+k_4 }{a^\ast_{k_1}a^\ast_{k_2}a_{k_3}a_{k_4}}.
\end{equation}
The symbol ``$\approx$" denotes the fact that we have approximated the two-body interaction potential with a delta function. 

A remark is in order. Namely, $\mathcal{H}$ as written in \eqref{Ham1} is manifestly particle conserving, and so its associated eigenvalue problem may be considered on the Fock space of fixed particle number, $\mathbb{F}_N$ for $N<\infty$. By contrast, the effective Hamiltonian of Lee Huang and Yang (which we denote by $\mathcal{H}_{\mathrm{LHY}}$ and define shortly), will not conserve the total number of particles. We will discuss the implications of this for the model of low-lying excitations after summarizing the remaining steps in the approximation of $\mathcal{H}_\mathrm{LHY}$. We proceed by decomposing the interaction part of $\mathcal{H}$ into terms containing like-powers of the condensate operators, $a_0$ and $a^\ast_0$~; the second, third, and fourth lines of the expression below contain quadratic, linear, and zeroth-order terms in these operators respectively:
\begin{equation}\label{approx2}
\begin{split}
\mathcal{H} &= \sum_{k\in\mathbb{Z}^3_L, k\not=0}{k^2 a_k^\ast a_k} + \frac{4\pi a }{|B_L|}(a_0^\ast)^2 a_0^2 \\
&+\frac{4\pi a}{|B_L|}\sum_{k\in\mathbb{Z}^3_L, k\not=0}{\Big((a_0^\ast)^2 a_k a_{-k} + a_k^\ast a_{-k}^\ast (a_0)^2 + 4 (a_0^\ast a_0) a_k^\ast a_k\Big)} \\
&+ \frac{8\pi a}{|B_L|}\sum\limits_{\substack{k_3=k_1+k_2\\ k_1,k_2,k_3\not=0 }}{\Big((a^\ast_{k_1}a^\ast_{k_2}a_{k_3})a_{0}+(a_{k_1}a_{k_2}a_{k_3}^\ast) a_0^\ast\Big)}\\
&+\frac{4\pi a}{|B_L|}\sum\limits_{\substack{k_1+k_2=k_3+k_4 \\k_1,k_2,k_3,k_4\not=0}}{a^\ast_{k_1}a^\ast_{k_2}a_{k_3}a_{k_4}}.
\end{split}
\end{equation}

The approximation of Lee Huang and Yang is consistent with the following steps: \textbf{(i)} replace the condensate occupation number operator by
$$ (a_0^\ast a_0) \mapsto N-\sum_{k\not=0}{a_k^\ast a_k},\quad\mathrm{for}\quad N<\infty,$$
and \textbf{(ii)} take the formal limit $N,|B_L|\to \infty$ with $\rho := N/|B_L|$ fixed.
In particular, the first line in \eqref{approx2} is approximated by
\begin{equation} \label{approx1}
\begin{split}
\frac{4\pi a}{|B_L|} (a_0^\ast)^2 a_0^2 &= \frac{4\pi a}{|B_L|}\big\{(a_0^\ast a_0)^2 - (a_0^\ast a_0)\big\} \\
&\approx \frac{4\pi a}{|B_L|}\Big\{(N-\sum_{k\not=0}{a^\ast_k a_k})^2-(N-\sum_{k\not=0}{a^\ast_k a_k})\Big\} \\
&\approx \frac{4\pi a}{|B_L|}\Big\{N(N-1)-2N\sum_{k\not=0}{a_k^\ast a_k}\Big\} \\
&\approx 4\pi a\rho N - 8\pi a\rho \Big\{\sum_{k\not=0}{a_k^\ast a_k}\Big\}.
\end{split}
\end{equation}
The third line of \eqref{approx1} makes use of step \textbf{(ii)} to drop the terms $\frac{4\pi a}{|B_L|}(\sum_{k}{a_k^\ast a_k})^2$ and $\frac{4\pi a}{|B_L|}\sum_{k\not=0}({a_k^\ast a_k})$ since these \textit{formally} vanish in the limit $|B_L|,N\to \infty$.

The second line of \eqref{approx2} consists of three quadratic terms in condensate operators, proportional to $(a_0)^2$, $(a_0^\ast)^2$ and $(a_0^\ast a_0)$ respectively. The $(a_0^\ast a_0)$ term will contribute to the diagonal part of $\mathcal{H}_\mathrm{LHY}$, via

\begin{equation}\label{approx3}\begin{split}
\frac{16\pi a}{|B_L|}\sum_{k\not=0}{(a_0^\ast a_0)a_k^\ast a_k} &\approx 16\pi a \rho \sum_{k\not=0}{a_k^\ast a_k} - \frac{16\pi a}{|B_L|}\big(\sum_{k\not=0}{a_k^\ast a_k}\big)^2 \\
&\approx 16\pi a\rho \sum_{k\not=0}{a_k^\ast a_k}.
\end{split}\end{equation}

If cubic and quartic terms in the $\{a^{(\ast)}_k\}$ operators where $k\not=0$ are now neglected (we will provide a consistent explanation for this in the next paragraph), approximations \eqref{approx1} and \eqref{approx3} yield the reduced Hamiltonian
\begin{equation}\label{H particle conserving}
\mathcal{H}\approx 4\pi a \rho N +\sum_{k\in\mathbb{Z}^3_L, k\not=0}{\big\{k^2 +8\pi a\rho \big\}a_k^\ast a_k} + \frac{4\pi a}{|B_L|}\sum_{k\in\mathbb{Z}^3_L, k\not=0}{\big\{(a_0^\ast)^2 a_k a_{-k} + a_k^\ast a_{-k}^\ast (a_0)^2\big\}}.
\end{equation}
This intermediate approximation to $\mathcal{H}_\mathrm{LHY}$ is the subject of section \eqref{PC Ham}. Since \eqref{H particle conserving} is particle conserving, the analysis of its spectrum is simpler than that of $\mathcal{H}_\mathrm{LHY}$. In particular, the non-Hermitian transform of \eqref{H particle conserving} will be bounded on the Fock space of fixed particle number, $\mathbb{F}_N$. 

It is now apparent that the only way for off-diagonal terms of \eqref{H particle conserving} to (formally) contribute in the limit described, that is, as $N,|B_L|\to\infty$ with $\rho = N/|B_L|$ fixed, requires that $(a_0)^2$ and $(a_0^\ast)^2$ be replaced (or be replaceable) by $N$. This also justifies the dropping of cubic-and-quartic terms in the $\{a^{(\ast)}_k\}$ operators between \eqref{approx2} and \eqref{H particle conserving}, since replacing $a_0$, $a_0^\ast$ by $\sqrt{N}$ and taking the above limit will result in these terms formally vanishing. We note that this amounts to a version of the famous Bogoliubov approximation \color{black}\cite{Bogoliubov}\color{black}, although such terminology was not used explicitly by Lee, Huang and Yang. With this final replacement, we arrive at the Lee-Huang-Yang (LHY) Hamiltonian
\begin{equation}
\mathcal{H}_{\mathrm{LHY}} := 4\pi a \rho N + \sum_{k\in\mathbb{Z}^3_L,\, k\not= 0}{\big\{k^2+8\pi a \rho\big\}a^\ast_k a_k} + 4\pi a \rho\sum_{k\in\mathbb{Z}^3_L,\, k\not= 0}{\big\{a^\ast_k a^\ast_{-k}+a_k a_{-k}\big\}}.
\end{equation}
The coupling between equal-and-opposite momenta present in $\mathcal{H}_\mathrm{LHY}$ makes it convenient to introduce the \textit{momentum half space}: $$\mathbb{Z}^+_L:=\{k= 2\pi n/L\,\big|\,n\in\mathbb{Z}^3,\,n_3>0\},$$ 
which allows us to write 
\begin{equation}\label{Heff}
\mathcal{H}_{\mathrm{LHY}} =  4\pi a \rho N + \sum_{k\in\mathbb{Z}^+_L}{2\big(k^2 + 8\pi a \rho\big)\Big\{\frac{1}{2}\big(a^\ast_k a_k + a^\ast_{-k}a_{-k}\big) + \frac{4\pi a \rho}{(k^2+8\pi a\rho)}\big(a^\ast_k a^\ast_{-k}+a_k a_{-k}\big)\Big\}}. \\
\end{equation}
Consistent with the work \color{black} \cite{Lee Huang Yang I}\color{black}, we define the constant 
$$y(k) := \frac{4\pi a\rho}{k^2+8\pi a\rho},\quad k\in\mathbb{Z}^+_L.$$ 

The point spectrum of $\mathcal{H}_\mathrm{LHY}$ will be called the \textit{Bogoliubov spectrum}, and is given by $\sum_{k\in\mathbb{Z}^+_{L}}{(n_k+n_{-k})\epsilon_k}$ where the set of occupation numbers $\{n_k\}_{k\in\mathbb{Z}_L}$ satisfies $$\sum_{k\in\mathbb{Z}^+_{L}}{(n_k+n_{-k})}<\infty$$ and the single particle energies are given by
\begin{equation}\label{BogSpec}
\epsilon_k := k\sqrt{k^2+16\pi a\rho}.
\end{equation}

In the analysis of the next section, it will be easier to state results for the generic operator that appears inside the brackets in \eqref{Heff}, where $y(k)$ may take on the range of values $(0,1/2)$ as $k$ varies over $\mathbb{Z}^+_L$. We therefore conclude this section by defining the quadratic Hamiltonian for two particle species $\mathcal{H}_\mathrm{ab}(y)$.
\begin{definition}
Let the operators $\{a^{(\ast)},b^{(\ast)}\}$ satisfy the commutation relations of any $\{a_k^{(\ast)},\, a_{-k}^{(\ast)}\}$ for $k\in\mathbb{Z}_L^+$, and act on the Fock space $\mathbb{F}_{ab}$ which is the linear span of all formal tensor products of the operators $a^\ast, b^\ast$, i.e., 
\begin{equation}
\mathbb{F}_{ab} := \mathrm{span}\{|n_a,n_b\rangle, n_a,n_b\in\mathbb{N}\},\quad\mathrm{for}\quad |n_a,n_b\rangle := \frac{(a^\ast)^{n_a}(b^\ast)^{n_b}}{\sqrt{n_a! n_b!}}|vac\rangle.
\end{equation}
Define the Hamiltonian $\mathcal{H}_{ab}(y)$ on $\mathbb{F}_{ab}$, where $0<y<1/2$, by 
\begin{equation}
\mathcal{H}_\mathrm{ab}(y) := \frac{1}{2}\big(a^\ast a+b^\ast b\big)+ y\big(a^\ast b^\ast + a b\big).
\end{equation}
We then define the Bogoliubov spectrum for $\mathcal{H}_\mathrm{ab}$ by the formula 
\begin{equation}\label{BogSpec2}
(n_a + n_b)\sqrt{1-4y^2}, \quad\mathrm{for}\quad n_a+n_b<\infty.
\end{equation}
\end{definition}
\subsection*{Comment on particle non-conservation} 
There is an apparent contradiction in the approximation scheme above, since step \textbf{(i)} is a statement of particle-conservation, while we ultimately arrive at $\mathcal{H}_\mathrm{LHY}$ which does not conserve the number of particles. We must therefore be more precise, and specify that the steps \textbf{(i)} and \textbf{(ii)} do not refer to the restriction of the problem to the $N$-particle sector of Fock space, $\mathbb{F}_N$, but rather impose a constraint on the average number of particles for the many-body states of $\mathcal{H}_\mathrm{LHY}$. This means that we solve the eigenvalue problem  $\mathcal{H}_\mathrm{LHY}|\Psi\rangle = E|\Psi\rangle$ for eigenstates $|\Psi\rangle\in\mathbb{F}$ which satisfy the condition 
\begin{equation}\label{avgParticles}
\langle \Psi | a_0^\ast a_0 + \sum_{k}{a_k^\ast a_k} |\Psi\rangle = N.
\end{equation}
We will have to verify \eqref{avgParticles} \textit{a posteriori} as a self-consistent assumption. 
\section{Formal construction of many-body excited states}\label{sec:construction}
We now construct a family of \textit{formal} solutions to the eigenvalue problem for Hamiltonian $\mathcal{H}_\mathrm{LHY}$. This construction makes use of the non-unitary transformation by $\exp(\mathcal{P}):\mathbb{F}\to\mathbb{F}$, which creates pairs of particles with opposite-momenta in every fiber $\mathbb{F}_n$ of $\mathbb{F}$. In defining the operator $\mathcal{P}$, we are inspired by the operator $\mathcal{W}$ described by Wu (equation \eqref{Wuexp}) \color{black}\cite{Wu61, Wu98}\color{black}. The operator $\mathcal{P}$ depends on a free parameter, denoted $0<\alpha(k)<1$, for every $k\in\mathbb{Z}^+_L$, and is defined by

\begin{equation}\label{Poperator}
\mathcal{P} := \sum_{k\in\mathbb{Z}^+_L}{-\alpha(k)a^\ast_k a^\ast_{-k}}.
\end{equation}

We will choose $\alpha(k)$ in the manner of Wu in order to eliminate all terms proportional to $(a^\ast_k a^\ast_{-k})$ in $\exp(\mathcal{P})(\mathcal{H}_\mathrm{LHY})\exp(-\mathcal{P})$. The resulting non-Hermitian eigenvalue problem can be solved exactly by considering solutions $|\Psi_\mathrm{LHY}(k)\rangle$, for $k\in\mathbb{Z}^{+}_L$, which are linear combinations of tensor product states containing only the momenta $k$ and $-k$:
\begin{equation}\label{transformedHab}
\Big(\exp(-\mathcal{P})\mathcal{H}_\mathrm{LHY}\exp(\mathcal{P})\Big) |\Psi_\mathrm{LHY}(k)\rangle = E|\Psi_\mathrm{LHY}(k)\rangle.
\end{equation} 
Equation \eqref{transformedHab} will admit either a discrete or a continuous point spectrum depending on whether the quantity 
$$ \tilde{y}(k) := \frac{y(k)}{\sqrt{1-4y(k)^2}},\quad \mathrm{for}\quad y(k) = \frac{4\pi a\rho}{k^2 + 8\pi a \rho},$$
satisfies $\tilde{y}(k)<1$ or $\tilde{y}(k)\ge1$. This complication is a mathematical artifact of the transformation; it has no significance for the physics problem. In the next section, we will provide a characterization of the states $\exp(-\mathcal{P})|\Psi_\mathrm{LHY}(k)\rangle$, where $|\Psi_\mathrm{LHY}(k)\rangle$ is an eigenstate of $\mathcal{H}_\mathrm{LHY}$. In a sense that will be made clear, the case $\tilde{y}(k)<1$ corresponds to a regular perturbation of the eigvenvalue problem for the diagonal Hamiltonian $\mathcal{H}_\mathrm{diag}:=\frac{1}{2}(a_k^\ast a_{k}+a_{-k}^\ast a_{-k})$, while $\tilde{y}(k)\ge1$ constitutes a singular perturbation. We note that for $y(k)\in(0,1/2)$, the quantity $\tilde{y}(k)$ takes on values in the range $(0,\infty)$. 

\subsection*{The non-Hermitian transform}
For $\mathcal{P}$ given by \eqref{Poperator}, a straightforward calculation using Lemma \eqref{lemma1} yields the conjugations 
\begin{equation}\label{conjugation0}
\exp(-{\mathcal{P}})a_k \exp({\mathcal{P}}) = a_k - \alpha(k) a^\ast_{-k},\quad \exp(-{\mathcal{P}}) a_k^\ast \exp({\mathcal{P}}) = a_k^\ast,\quad \forall k\in\mathbb{Z}_L^3.
\end{equation}
We define $\alpha(-k)=\alpha(k)$ for $k\in\mathbb{Z}^+_L$. This yields the transformed Hamiltonian
\begin{equation}\label{eq:expTransform} \begin{split}
\exp(-{\mathcal{P}})\mathcal{H}_\mathrm{LHY}&\exp({\mathcal{P}}) = 4\pi a \rho N +4\pi a\rho \sum_{k\in\mathbb{Z}^3_L}{\alpha(k)} \\
&+\sum_{k\in\mathbb{Z}^3_L}{\big(k^2+8\pi a\rho - 8\pi a\rho \alpha(k)\big)(a_k^\ast a_k)} + 4\pi a\rho\sum_{k\in\mathbb{Z}_L^3}{(a_k a_{-k})} \\
&+\sum_{k\in\mathbb{Z}^3_L}{\big(4\pi a\rho (\alpha(k))^2 - (k^2 +8\pi a\rho)\alpha(k) +4\pi a\rho\big) (a_k^\ast a_{-k}^\ast)}.
\end{split}
\end{equation}
We note that at this point that the term $\sum_{k\in\mathbb{Z}^3_L}{\alpha(k)}$ introduces an infinite constant (see equation \eqref{alphak}) to the energy spectrum of the transformed Hamiltonian; this can be attributed to approximating the smooth interaction potential by a delta function and can be removed in a systematic way. We will keep the term as an infinte constant in the following analysis, and instead refer to \color{black}\cite{Huang Yang Luttinger}\color{black}~ for details regarding the renormalization procedure. In distinction to the unitary transformation of $\mathcal{H}_\mathrm{LHY}$ exemplified by equation \eqref{unitary transformation}, the exponential transformation by $\mathcal{P}$ has a free parameter $\alpha(k)$ for every $k\in\mathbb{Z}^+_L$ which is not constrained by a requirement of uniticity. We exploit this freedom in order to make the expansion of eigenvectors in the momentum basis particularly simple. The resulting problem \eqref{transformedHab} on Fock space will correspond to an upper triangular (infinite) matrix system. 

In this vein, we choose $\alpha(k)$ so that the last line of \eqref{eq:expTransform} vanishes, which implies 
\begin{equation}\label{alphak}
\alpha(k) = \frac{1}{8\pi a\rho}\Big((k^2 + 8\pi a\rho)\pm k\sqrt{k^2 + 16\pi a\rho}\Big),\quad k\in\mathbb{Z}^+_L.
\end{equation}
The choice of $\alpha(k)$ corresponding to the minus sign in \eqref{alphak} yields a positive spectrum for \eqref{transformedHab}. The two possible solutions for $\alpha(k)$ in the expansion of $\mathcal{P}$ has an analog in the non-periodic setting, where the quadratic equation for $\alpha(k)$ generalizes to an operator Riccati equation for a pair excitation kernel $k(x,y)$. In \color{black}\cite{Grillakis Margetis Sorokanich}\color{black}~, we provide a detailed description of this correspondence. With this choice, the transformed Hamiltonian reads


\begin{equation}\label{eq:expTransform2}\begin{split}
\exp(&-\mathcal{P})\mathcal{H}_\mathrm{LHY}\exp({\mathcal{P}}) = 4\pi a\rho N + 4\pi a\rho \sum_{k\in\mathbb{Z}^3_L}{\alpha(k)} \\ 
&+ \sum_{k\in\mathbb{Z}^+_L}{2\big(k\sqrt{k^2+16\pi a\rho}\big)\Big\{\frac{1}{2}(a^\ast_k a_k + a^\ast_{-k}a_{-k})+ \Big(\frac{4\pi a\rho}{k\sqrt{k^2+16\pi a\rho}}\Big) a_k a_{-k}\Big\}}.
\end{split}\end{equation}

For a single term of this sum corresponding to fixed $k\in\mathbb{Z}^+_L$, the operator inside brackets is a transformation of the formal Hamiltonian $\mathcal{H}_\mathrm{ab}(y):\mathbb{F}_\mathrm{ab}\to\mathbb{F}_\mathrm{ab}$ if we equate momentum $k$ with the particle species $a$ and $-k$ with particle species $b$. This is summarized in the following definition:

\begin{definition}\label{Def2} Denote the rescaled transformed Hamiltonian 
\begin{equation}\begin{split}
\mathcal{H}^{(\alpha)}_\mathrm{ab} &:= \frac{1}{1-2\alpha y}\Big(\exp(\alpha a^\ast b^\ast) \mathcal{H}_\mathrm{ab} \exp(-\alpha a^\ast b^\ast) + \alpha y\Big) \\
&= \frac{1}{2}\big(a^\ast a + b^\ast b\big) + y_1(\alpha) ab + y_2(\alpha) a^\ast b^\ast,
\end{split}
\end{equation}
where the parameters $y_1(\alpha),\,y_2(\alpha)$ are given by $$y_1(\alpha):=\frac{y}{1-2\alpha y},\quad \mathrm{and}\quad y_2(\alpha):=\frac{y-\alpha+\alpha^2y}{1-2\alpha y}.$$
The value $$\alpha_c := \frac{1-\sqrt{1-4y^2}}{2y},$$ then induces the transformation:
\begin{equation}\label{eq:WuHamiltonian1}
\mathcal{H}^{(\alpha_c)}_\mathrm{ab}=\frac{1}{2}\Big(a^*a+b^*b\Big)+\Big(\frac{y}{\sqrt{1-4y^2}}\Big)ab.
\end{equation}
If $y = y(k) = \frac{4\pi a\rho}{k^2 + 8\pi a\rho}$, the operator \eqref{eq:WuHamiltonian1} acts on $\mathbb{F}_\mathrm{ab}$ in  precisely the same way as the operator inside brackets in equation \eqref{eq:expTransform2} acts on the Fock subspace formed by tensor product states with momentum $k,-k$. 
\end{definition}


\subsection*{Constructing eigenstates}
\label{subsec:critical}

The formal eigenstates of $\mathcal{H}^{(\alpha_c)}_\mathrm{ab}$ are now described. First we explain the setup. By \textit{formal}, we mean that formula \eqref{formula1} does not specify whether the states described there have finite norm in the space $\mathbb{F}_\mathrm{ab}$. We nonetheless write formula \eqref{formula1}, abusing notation, since it neatly encompasses all of the cases involved in the spectral analysis of the family of operators $\mathcal{H}^{(\alpha_c)}_\mathrm{ab}(y)$. The states are written $|\Psi_{p,\Theta}\rangle\in\mathbb{F}_\mathrm{ab}$, for $p\in\mathbb{N}$ and $\Theta\in\mathbb{C}$, and satisfy the formal eigenvalue problem 
\begin{equation}\label{ab-eigenproblem}
\mathcal{H}^{(\alpha_c)}_\mathrm{ab}|\Psi_{p,\Theta}\rangle = \big(\frac{p}{2}+\Theta\big) |\Psi_{p,\Theta}\rangle.
\end{equation}
The index $p$ describes the fact that $|\Psi_{p,\Theta}\rangle$ is an element of the linear subspace of states $|\Psi\rangle\in\mathbb{F}_\mathrm{ab}$ such that $(a^\ast a - b^\ast b)|\Psi\rangle = p |\Psi\rangle$~ holds. In the occupation number basis, this subspace corresponds to the closed linear span
$$|\Psi_{p,\Theta}\rangle\in\mathrm{span}\big\{|p+s,s\rangle, s\in\mathbb{N}\big\}.$$
For every $p\in\mathbb{N}$, it is also possible to consider solutions which satisfy $(b^\ast b - a^\ast a)|\Psi\rangle = p|\Psi\rangle$, and construct eigenstates in the span of vectors $|s,s+p\rangle$. Thus every state $|\Psi_{p,\Theta}\rangle$ will be associated with a degenerate state denoted $|\Psi^{(-)}_{p,\Theta}\rangle$. Indeed, the only difference between $|\Psi_{p,\Theta}\rangle$ and $|\Psi_{p,\Theta}^{(-)}\rangle$ will be whether the expansion is formed from states $|s+p,s\rangle$ versus $|s,p+s\rangle$. We therefore write formulas explicitly for the first kind of eigenstate, implying that equivalent formulas can be derived for $|\Psi^{(-)}_{p,\Theta}\rangle$ by the appropriate substitution.

The complex value $\Theta$ enters into the expansion for the state $|\Psi_{p,\Theta}\rangle$ via the generalized binomial coefficients $\binom{\Theta}{s}$ for $s=0,1,2,\dots$. The (unnormalized) states $|\Psi_{p,\Theta}\rangle$ are now described:

\begin{formula}\label{formula1} Let $p\in\mathbb{N}$ and $\Theta\in \mathbb{C}$, and define $\tilde{y}$ by 
$$\tilde{y} := \frac{y}{\sqrt{1-4y^2}},\quad\mathrm{for}\quad 0<y<1/2.$$
In the occupation number basis for the particle species $a,b$, define the formal expansion
\begin{equation}\label{eq:GenWuState0}
|\Psi_{p,\Theta}\rangle:= \sum_{s=0}^\infty{(\tilde{y})^{-s}\binom{\Theta}{s}\sqrt{\binom{p+s}{s}}^{-1}}|p+s,s\rangle.
\end{equation}
It can be verified by direct computation that $|\Psi_{p,\Theta}\rangle$ satisfies the eigenvalue problem \eqref{ab-eigenproblem} with energy $E=\frac{p}{2}+\Theta$. A second collection of states, denoted $|\Psi^{(-)}_{p,\Theta}\rangle$, are constructed in an identical way, replacing $|p+s,s\rangle$ by $|s,p+s\rangle$ in \eqref{eq:GenWuState0}. A special case occurs when $\Theta = N\in \mathbb{N}$. The resulting states are then finite linear combinations, and formula \eqref{eq:GenWuState0} is replaced by
\begin{equation}\label{eq:GenWuState1}
|\Psi_{p,N}\rangle:= \sum_{s=0}^N{(\tilde{y})^{-s}\binom{N}{s}\sqrt{\binom{p+s}{s}}^{-1}}|p+s,s\rangle.
\end{equation}
\end{formula}


\subsubsection*{Comments on the construction of eigenstates}
The derivation of formula \eqref{eq:GenWuState0} is elaborated in the appendix, and follows from a direct expansion in the particle number basis. Some key points are now emphasized: 
\begin{enumerate}
\item{} All eigenstates of $\mathcal{H}^{(\alpha_c)}_\mathrm{ab}$ (as well as $\mathcal{H}_\mathrm{ab}$) must be superpositions of states with fixed momentum $p$, either given by $|p+s,s\rangle$, or $|s,p+s\rangle$ for $s\in\mathbb{N}$. This is because the operator (Casimir's operator \color{black}\cite{Perelomov}\color{black})
\begin{equation} \mathcal{C} := \frac{1}{4}-\frac{1}{4}(a^\ast a -b^\ast b)^2 \end{equation}
commutes with $a^\ast b^\ast$, $ab$, and $\frac{1}{2}(a^\ast a+b^\ast b+1)$. Casimir's operator is constant on each of the subspaces
$\mathrm{span}\{|p+s,s\rangle|\,\,s\in\mathbb{N}\}\subset\mathbb{F}_\mathrm{ab}$ for $p\in\mathbb{N}$.

\item{} Formula \eqref{formula1} suggests that every complex number is a (possible) eigenvalue of $\mathcal{H}^{(\alpha_c)}_\mathrm{ab}$. This is sometimes, but not always, the case. In particular, for $\tilde{y}<1$, the only states in the spectrum will be $|\Psi_{p,N}\rangle$ for $N\in\mathbb{N}$, while for $1\le\tilde{y}$, the states $|\Psi_{p,\Theta}\rangle$ are eigenstates for all $\Theta\in\mathbb{C}$. It is clear that the coefficients in \eqref{eq:GenWuState0} can be approximated by $\tilde{y}^{-s}\binom{\Theta}{s}$ for large $s$, so the parameters $\tilde{y}, \Theta$ will determine which of the Fock space vectors $|\Psi_{p,\Theta}\rangle$ are normalizable. 

\item{} When $\Theta=N\in\mathbb{N}$, the expansion \eqref{eq:GenWuState0} is a finite sum, and the energy of state $|\Psi_{p,N}\rangle$ is $E = N+\frac{p}{2}$. The states $|\Psi_{p,N}\rangle$ have the same energies as the tensor product eigenstates of the Hamiltonian $ \frac{1}{2}(a^\ast a + b^\ast b)$, which is the limit of the operator $\mathcal{H}^{(\alpha_c)}_\mathrm{ab}$ as $y\to 0$. 

\end{enumerate}

The remainder of this work is devoted to analyzing the spectrum of $\mathcal{H}^{(\alpha_c)}_\mathrm{ab}$ in the two regimes $\tilde{y}<1$ and $1\le\tilde{y}$. We undertake this in order to answer the physically pertinent questions: \textbf{(i)} which of the states $|\Psi_{p,\Theta}\rangle$ in formula \eqref{formula1} determine eigenstates $|\Psi_E\rangle$ of $\mathcal{H}_\mathrm{ab}$ by means of the transform 
\begin{equation}
|\Psi_E\rangle = \exp(-\alpha_c a^\ast b^\ast)|\Psi_{p,\Theta}\rangle\,\,?
\end{equation}
\textbf{(ii)} If the point spectrum of $\mathcal{H}^{(\alpha_c)}_\mathrm{ab}$ is the entire complex plane, is there an easy way to distinguish the mappings of states $|\Psi_E\rangle$ among the states $|\Psi_{p,\Theta}\rangle$? \textbf{(iii)} Finally, can all of the LHY eigenstates be recovered in this way? 

As aluded to, we will handle the cases $\tilde{y}<1$ and $\tilde{y}\ge 1$ separately. Questions \textbf{(i)}~\textendash~\textbf{(iii)} are answered definitively for $\tilde{y}<1$ in the next theorem, which provides a template for the more complicated case $\tilde{y}\ge 1$. When $\tilde{y}<1$, only the states $\{\Psi_{p,N}\}_{p,N\in\mathbb{N}}$ will be normalizable in $\mathbb{F}_\mathrm{ab}$ and in this sense the eigenvalue problem for $\mathcal{H}^{(\alpha_c)}_\mathrm{ab}$ is a regular perturbation of the problem for $\frac{1}{2}(a^\ast a+ b^\ast b)$. Sections \eqref{sec:non-uni-family} and \eqref{completeness} are devoted to the case $\tilde{y}\ge 1$.

\subsubsection*{Discrete spectrum in the Case $0<\tilde{y}<1$:} When $\tilde{y}<1$, the eigenstates of $\mathcal{H}^{(\alpha_c)}_\mathrm{ab}$ will be $|\Psi_{p,N}\rangle$ of formula \eqref{formula1}, where $N\in\mathbb{N}$. Energies of these states will be equal to the energies of tensor products $|p+N,N\rangle$ as eigenstates of $\frac{1}{2}(a^\ast a +b^\ast b)$. The condition $0<\tilde{y}<1$ in $\mathcal{H}^{(\alpha_c)}_\mathrm{ab}$ translates to the condition $|k|^2 > 4\pi a \rho(\sqrt{5}-2)>0$ in the $k-$momentum component of $\mathcal{H}_\mathrm{LHY}$. Thus, this case represents an infinite collection of momenta in the Lee-Huang-Yang Hamiltonian.

\begin{proposition} \label{prop1}
Suppose $0<\tilde{y}<1$. Then 
\begin{equation}
\sigma_\mathrm{p}\big(\mathcal{H}^{(\alpha_c)}_\mathrm{ab}\big) = \big\{\frac{p}{2}+N, \,\mathrm{for}\,\, p,N\in\mathbb{N}\big\}.
\end{equation}
The eigenstates are given by $|\Psi_{p,N}\rangle$, $|\Psi^{(-)}_{p,N}\rangle$ of formula \eqref{formula1}, for $N\in\mathbb{N}$, and exhibit the same degeneracy as the momentum states $\big\{|N,N+p\rangle\,\,\mathrm{or}\,\, |N+p,N\rangle\big\}$ as eigenstates of the operator $\frac{1}{2}(a^\ast a+b^\ast b)$ for $p,N\in\mathbb{N}.$
\end{proposition}
\begin{proof}

It is clear that the states $|\Psi_{p,N}\rangle$ with $N\in\mathbb{N}$ have finite norm in $\mathbb{F}_\mathrm{ab}$ as finite superpositions in the particle number basis. 

Now suppose $\Theta\in\mathbb{C}\setminus\mathbb{N}$, so the coefficients
$$c_s:=(\tilde{y})^{-s}\binom{\Theta}{s}\sqrt{\binom{p+s}{s}}^{\,\,-1},$$
are nonzero for all $s\ge0$. We proceed to show that $\big\{c_{s}\big\}_{s=0}^\infty$ cannot be square summable.

Indeed, using the Gamma function to describe the generalized binomial coefficient, we have
\begin{align*}
|c_{s}|^2
&=\tilde{y}^{-2s}\frac{\Gamma(1+p)}{\Gamma(-\Theta)^2}\cdot\frac{\Gamma(s-\Theta)^2}{\Gamma(s+p+1)\Gamma(s+1)}|c_{0}|^2. 
\end{align*}
Let us write \textit{Stirling's approximation} as follows \color{black} \cite{Erdelyi}\color{black}
\begin{equation}
\log(\Gamma(z)) \approx \big(z-\frac{1}{2}\big)\log(z) - z +\frac{1}{2}\log(2\pi),\quad \mathrm{as}\quad |z|\to\infty.
\end{equation}
Here, the symbol ``$\approx$" denotes asymptotic equivalence. Thus
\begin{equation}
\log\big(\Gamma(s-\Theta)\big) 
\approx s\log(s) -(\Theta+1/2)\log(s) -s + C,\quad \mathrm{as}\quad s\to\infty,
\end{equation}
where $C$ is a constant whose specific value does not matter. Similarly
\begin{equation}
\log(\Gamma(s+p+1)) \approx s\log(s) +(p+1/2)\log(s) -s,\quad\mathrm{as}\quad s\to\infty.
\end{equation}
Putting the two lines above together gives
\begin{equation}
\frac{\Gamma(s-\Theta)}{\Gamma(s+p+1)\Gamma(s+1)}\approx\frac{e^{2\log s}e^{-2(\Theta+1/2)\log(s)}e^{-2s}}{e^{s\log s}e^{(p+1/2)\log s}e^{- s}e^{s\ln s}e^{1/2\log s}e^{- s}} \\
=\frac{1}{s^{2\Theta+2+p}},\quad\mathrm{as}\quad s\to\infty.
\end{equation}
Therefore
\begin{equation}\label{eq:approxCoeff}
|c_{p+s,s}|^2\approx K\frac{\tilde{y}^{-2s}}{s^{2\Theta+2+p}}|c_{p,0}|^2,\quad\mathrm{for}\quad K=\frac{\Gamma(1+p)}{\Gamma(-\Theta)^2}\quad\mathrm{as}\quad s\to\infty.
\end{equation}
If $|\tilde{y}|<1$, then it follows that $\sum_{s=0}^\infty{|c_{s}|^2}$ cannot be finite. This concludes the proof. \color{black}
\end{proof}

When $\tilde{y}\ge 1$, the estimates in the above proof show that the states $|\Psi_{p,\Theta}\rangle$ are normalizable for all values $\Theta\in\mathbb{C}$. Answering the questions \textbf{(i)}~\textendash~\textbf{(iii)} becomes more difficult in this case. This is the subject of the next section.

\subsection*{A Geometric Picture of Degeneracies:} We conclude this section with a geometric description of the degenerate eigenspaces of $\mathcal{H}^{(\alpha_c)}_\mathrm{ab}$ in $\mathbb{F}_{ab}$ when $\tilde{y}<1$. The Fock space $\mathbb{F}_{ab}$ is the orthogonal direct sum of the following subspaces: 
$$\mathbb{F}_{ab}=\bigoplus_{p=-\infty}^{\infty}{\mathbb{F}^{(p)}},\quad\mathbb{F}^{(+p)}:=\mathrm{span}\{|p+s,s\rangle:s\in\mathbb{N}\},\quad\mathbb{F}^{(-p)}:=\mathrm{span}\{|s,p+s\rangle:s\in\mathbb{N}\}.$$
On the two-dimensional lattice describing occupation numbers $(n_a,n_b)$ of the states $|n_a,n_b\rangle$, the subspace $\mathbb{F}^{(p)}$ describes an infinite line segment terminating at a point $(0,p)$ or $(p,0)$. The states $|\Psi_{p,N}\rangle$ are elements of $\mathbb{F}^{(+p)}$, while the states $|\Psi^{(-)}_{p,N}\rangle$ are elements of $\mathbb{F}^{(-p)}$.

It was established that the state $|\Psi_{p,N}\rangle$ has energy $E=\frac{p}{2}+N$, identical to the energy of the tensor product $|p+N,N\rangle$ as an eigenstate of $\frac{1}{2}(a^\ast a+b^\ast b)$, and that within the subspace $\mathbb{F}^{(p)}$ it is the unique state with this energy. Accordingly, on $\mathbb{N}\times\mathbb{N}$, the line segment $\{(p,N)\big|\,\,\frac{p}{2}+N =E\}$ connecting endpoints $(p+N,N)$ and $(N,p+N)$ intersects all ordered pairs corresponding to degenerate states $|\Psi_{p,N}\rangle$ with energy $E$ \color{black}(see figure \eqref{fig1})\color{black}. 

Finally, consider the (singular) limit $\tilde{y}\to 0$ in the state $|\Psi_{p,N}\rangle$ with energy $E=\frac{p}{2}+N$ given by equation \eqref{eq:GenWuState1}. In this limit, the finite sum corresponding to the (normalized) state reduces to a single term
$$|\Psi_{p,N}\rangle\to|N+s,s\rangle,\quad\mathrm{as} \,\,\tilde{y}\to0.$$ This is expected by the fact that $\mathcal{H}^{(\alpha_c)}_\mathrm{ab}$ approaches $\frac{1}{2}(a^\ast a + b^\ast b)$ in this limit.

\begin{figure}[b!]
  \includegraphics[width=\linewidth]{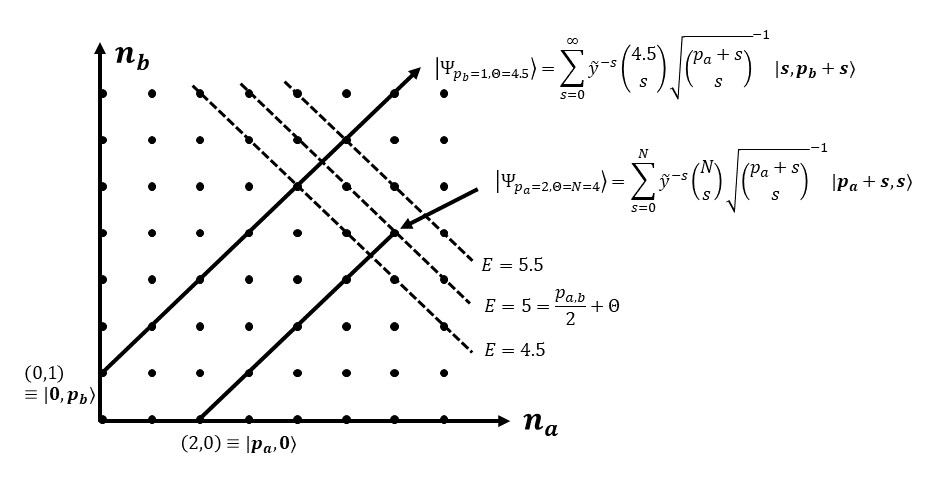}
  \caption{\color{black}\color{black}}
  \label{fig1}
\end{figure} 

\section{Extracting the Bogoliubov spectrum for $\tilde{y}\ge 1$}\label{sec:non-uni-family}
The current section contains the technical results which allows us to answer questions \textbf{(i)}~\textendash~\textbf{(iii)} of the previous section for the case $\tilde{y}\ge 1$. In the broader scope of the LHY Hamiltonian, this amounts to characterizing those solutions to the the eigenvalue problem 
$$\big(\exp(-\mathcal{P})\mathcal{H}_\mathrm{LHY}\exp(\mathcal{P})\big)|\Psi(k)\rangle = E |\Psi(k)\rangle,$$
which yield solutions to the eigenvalue problem
$$\mathcal{H}_\mathrm{LHY}|\tilde{\Psi}(k)\rangle = E|\tilde{\Psi}(k)\rangle$$ by means of the transform $\exp(\mathcal{P})|\Psi(k)\rangle = |\tilde{\Psi}(k)\rangle$. We do this by studying eigenstates of the formal Hamiltonian $\mathcal{H}^{(\alpha_c)}_\mathrm{ab}(\tilde{y})$ as in the previous section.  Proposition \eqref{prop1}  demonstrated that when $\tilde{y}\ge 1$, the eigenstates of $\mathcal{H}^{(\alpha_c)}_\mathrm{ab}(\tilde{y})$ are givn by $|\Psi_{p,\Theta}\rangle$ for all $p\in \mathbb{N}$ and $\Theta\in\mathbb{C}$. This poses an apparent problem for using the non-unitary scheme.

We proceed by considering the operator family $\mathcal{H}_\mathrm{ab}^{(\alpha)}$ for $0<\alpha\le\alpha_c$. When $\alpha = \alpha_c$, the formula
$$\mathcal{H}_\mathrm{ab}\big(e^{-\alpha_c a^\ast b^\ast}|\Psi_{p,\Theta}\rangle\big) = \Big((1+2\alpha_c y)(\frac{p}{2}+\Theta)+\alpha_c y\Big)(e^{-\alpha_c a^\ast b^\ast}|\Psi_{p,\Theta}\rangle),\quad p\in\mathbb{N},\,\Theta \in\mathbb{C},$$
implies that states in the set $\{e^{-\alpha_c a^\ast b^\ast}|\Psi_{p,\Theta}\rangle\}_{p,\Theta}$ are also eigenstates of $\mathcal{H}_\mathrm{ab}$, provided that $|\Psi_{p,\Theta}\rangle\in\mathrm{dom}\big(e^{-\alpha_c a^\ast b^\ast}\big)$. More specifically, given $\tilde{y}\ge 1$ we will show: \textbf{(a)} that the only states $|\Psi_{p,\Theta}\rangle\in\mathrm{dom}\big(e^{-\alpha_c a^\ast b^\ast}\big)$ occur for $\Theta\in\mathbb{N}$, and \textbf{(b)} that the collection $\{|\Psi_{p,N}\rangle\}_{p,N\in\mathbb{N}}$ constitutes a basis of $\mathbb{F}_\mathrm{ab}$. We conduct the first of these two tasks in the present section.\color{black}

\subsection*{Method of generating functions}
\label{sec:non-uni-family}
Inspired by Lee, Huang, and Yang \color{black}\cite{Lee Huang Yang I}\color{black}, we define the generating function $G_{|\Psi_E\rangle}(z):\mathbb{C}\to\mathbb{C}$ for eigenstate $|\Psi_E\rangle$ of $\mathcal{H}^{(\alpha)}_\mathrm{ab}$, where $0<\alpha\le \alpha_c$, and demonstrate how the singularities of $G_{|\Psi_E\rangle}$ in the complex plane furnish a criteria for determining eigenstates of the Hamiltonian $\mathcal{H}_\mathrm{ab}= \mathcal{H}^{(\alpha = 0)}_\mathrm{ab}$. More specifically, the action of the exponential map 
$$|\Psi_E\rangle\mapsto e^{-\alpha a^*b^*}|\Psi_E\rangle$$
induces the following transformation of the generating function (see proposition \eqref{Transform}):
$$G_{|\Psi_E\rangle}(z)\mapsto G_{e^{-\alpha a^*b^*}|\Psi_E\rangle}(z) = \Big(\frac{1}{1+\alpha z}\Big)G_{|\Psi_E\rangle}\Big(\frac{z}{1+\alpha z}\Big).$$
This fact is exploited for the case $\tilde{y}\ge 1$; the function $\frac{1}{1+\alpha z}G_\Psi(\frac{z}{1+\alpha z})$ must not have a  singularity in the unit disk if $|\Psi_E\rangle\in\mathrm{dom}(e^{-\alpha a^\ast b^\ast})$, which will only be possible for discrete energies $E$ corresponding to the Bogoliubov spectrum.

When $0<\alpha<\alpha_c(y)$, both of the constants $y_1(\alpha)$ and $y_2(\alpha)$ in the transformed Hamiltonian $\mathcal{H}^{(\alpha)}_\mathrm{ab}$ will be nonzero, and must be related to each other as in Definition \eqref{Def2}. We consider the range of values~ $0\le y_1(\alpha)<\infty$ ~as ~$0\le \alpha_c<1$ when $k\in\mathbb{Z}^+_L$. Eigenstates of $\mathcal{H}^{(\alpha)}_\mathrm{ab}$ are no longer finite superpositions in the momentum basis. The Casimir operator $\mathcal{C}$ still commutes with $\mathcal{H}^{(\alpha)}_\mathrm{ab}$, so we know that eigenstates of $\mathcal{H}^{(\alpha)}_\mathrm{ab}$ will remain superpositions of the states $|s+p,s\rangle$ or $|s,s+p\rangle$ for fixed $p$.

\begin{definition} \label{genfunc} Let $|\Psi\rangle\in\mathbb{F}_\mathrm{ab}$ have expansion in the particle number basis given by  
$$|\Psi\rangle=\sum_{s=0}^\infty{c_s|p+s,s\rangle},\quad\mathrm{for}\quad \{c_{s}\}_{s=0}^\infty\in\ell^2(\mathbb{C}),$$
and define the rescaled coefficients $C_s:=\sqrt{\frac{s!}{(p+s)!}}c_s$ (alternatively $|\Psi\rangle$ can have an exapnsion in the states $|s,s+p\rangle$).
The generating function which corresponds to state $|\Psi\rangle$ is a formal power series in complex variable $z$ defined by  $$G_{|\Psi\rangle}(z):=\sum_{s=0}^\infty{C_sz^s}.$$
\end{definition}

It is clear that for fixed $p\in\mathbb{N}$, the sequence $\{C_s\}_{s=0}^\infty$ is square summable if and only if $\{c_s\}_{s=0}^\infty$ is as well. Since $$|\sum_s{C_sz^s}|^2\le \sum_s{|C_s|^2}\sum_s{|z|^{2s}},$$ the condition that $\||\Psi\rangle\|_{\mathbb{F}_\mathrm{ab}}=1$ means that $G_{|\Psi\rangle}(z)$ will be analytic in the unit disk, $\{z\in\mathbb{C}:|z|<1\}$.


\color{black}We refer to the appendix for details regarding the following proposition, which describes properties of the function $G_{|\Psi_E\rangle}(z)$ when $|\Psi_E\rangle$ is an eigenstate of $\mathcal{H}^{(\alpha)}_\mathrm{ab}$.\color{black}

\begin{proposition}\label{genfuncfacts}
Let $|\Psi_E\rangle\in\mathbb{F}_\mathrm{ab}$ be an eigenstate of $\mathcal{H}^{(\alpha)}_\mathrm{ab}$ for $0<\alpha\le\alpha_c\le1$, with energy $E\in\mathbb{C}$, and define the generating function $G_{|\Psi_E\rangle}(z)$ according to definition \eqref{genfunc}. Then
\begin{enumerate}

\item{} $G_{|\Psi_E\rangle}(z)$ satisfies the first-order ordinary differential equation:
\begin{equation}\label{eq:diffEq}
\begin{split}
z &\Big\{y_2(\alpha)z^2+z+y_1(\alpha)\Big\}\big(G_{|\Psi_E\rangle}(z)\big)' \\ +&\Big\{y_2(\alpha)z^2+\left(\frac{p}{2}-E\right)z+y_1(\alpha)p\Big\}G_{|\Psi_E\rangle}(z)=C_0y_1(\alpha)p,
\end{split}\end{equation}
where $C_0$ is the zeroth rescaled coefficient in the definition of $G_{|\Psi_E\rangle}(z)$.

\item{} The general solution to \eqref{eq:diffEq} is given by
\begin{equation}
G_{|\Psi_E\rangle}(z) = G_{hom}(\Psi_E,z) + I_{\Psi_E}(z),
\end{equation}
where the solution to the homogeneous problem reads
\begin{equation}\label{eq:homogenous}
G_{hom}(\Psi_E,z)=Kz^{-p}(z-z_+)^B(z-z_-)^C,
\end{equation}
for  $K,B\in\mathbb{C}$ and $C=p-1-B$, and 
\begin{equation}
I_{\Psi_E}(z):= \frac{C_0y_1(\alpha)p}{y_2(\alpha)}(z-z_+)^B(z-z_-)^C\int_0^z{u^{p-1}(u-z_+)^{-(1+B)}(u-z_-)^{-(1+C)}du}.
\end{equation}
The integral in this expression is a contour integral with contour lying within a simply-connected region of analyticity so that it is uniquely specified by the endpoints. 

\item{} The constant $K$ is arbitrary, while the constant $B$ is related to the energy $E\in\mathbb{C}$ via 
\begin{equation}\label{eq:BandE}
B
=\frac{(p/2-E-1)z_{+}(\alpha)+y_1(\alpha)(p-1)}{z_+(\alpha)+2y_1(\alpha)}.
\end{equation}

\item{} For $p>0$, $G_{hom}(\Psi_E,z)$ has a pole at $z=0$, but $I_{\Psi_E}(z)$ does not. It follows that $G_{|\Psi_E\rangle}(z) = I_{\Psi_E}(z)$. The function $I_{\Psi_E}(z)$ manifests a possible singularity at $z=z_+$ depending on the value of $B$. 
\end{enumerate}
\end{proposition}

It is well-known that singularities of solutions to \eqref{eq:diffEq} can occur only at roots of the leading polynomial $Q(z) := z\big(y_2(\alpha)z^2+z+y_1(\alpha)\big)$, which are given by $z = \{0,z_\pm(\alpha)\}$ for
\begin{equation}\label{eq:zeros}
z_\pm(\alpha):=\frac{-1\pm\sqrt{1-4y_1(\alpha)y_2(\alpha)}}{2y_2(\alpha)},\quad\mathrm{assuming}\quad y_2(\alpha)\not=0.
\end{equation}
Let us now state a few properties of the zeros $z_{\pm}(\alpha)$. Recall that $\alpha_c = \frac{1-\sqrt{1-4y^2}}{2y}$:
\begin{itemize}
\item{} For any $\alpha\in[0,\alpha_c]$, we have $1-2\alpha y >0$.
\item{} For any $\alpha\in[0,\alpha_c]$, the polynomial $P(\alpha) = y-\alpha +\alpha^2 y$ is strictly positive. Hence, both $y_1(\alpha),\,y_2(\alpha)$ given in definition \eqref{Def2} are strictly positive for $\alpha\in[0,\alpha_c]$.
\item{} The discriminant, $1-4y_1(\alpha)y_2(\alpha)$ of equation \eqref{eq:zeros} is strictly positive for any $\alpha\in[0,\alpha_c]$. Hence, $z_\pm(\alpha)$ are real for $\alpha$ in the prescribed range.
\item{} $z_+(\alpha)z_-(\alpha)= 1$ with $|z_+(0)|<1$. Hence, $|z_-(0)|>1$. In addition, $|z_-(\alpha)|>1$ for any $\alpha\in\big(0,\alpha_c\big)$, with $z_-(\alpha)\to -\infty$ as $\alpha\uparrow \alpha_c$.
\item{} By contrast, $z_+(\alpha)$ may or may not be inside the unit circle. Indeed, the condition $|z_+(\alpha)|<1$ implies $y_1(\alpha)+y_2(\alpha)<1$, which is true for all $\alpha\in[0,\alpha_c]$ only if $$\tilde{y} = \frac{y}{\sqrt{1-4y^2}}<1.$$ 
Otherwise there exists some $\alpha_{max}<\alpha_c$ such that $|z_+(\alpha)|<1$ for $\alpha<\alpha_{max}$ and $|z_+(\alpha)|\ge 1$ for $\alpha_{max} < \alpha < \alpha_c$. We handle these two cases separately.
\end{itemize}

\subsubsection*{The Bogoliubov spectrum for $|z_+|<1$:}

The point spectrum $\sigma_\mathrm{p}(\mathcal{H}_\mathrm{ab})$ is found via values of $B$ which render $G_{|\Psi_E\rangle}(z)$ analytic in the unit disk using equation \eqref{eq:BandE}. As an example, consider the case $\alpha\le \alpha_c$ and $p=0$, so that $$G_{|\Psi_E\rangle}(z)=G_{hom}(z)=K\big(z-z_+(\alpha)\big)^B\big(z-z_-(\alpha)\big)^{-1-B}.$$ Since $z_+(\alpha)$ lies inside the unit disk by assumption, $G_{|\Psi_E\rangle}(z)$ is analytic in this region \textit{if and only if} $B\in\mathbb{N}$. It follows that the allowed energies $E$ must be discrete. This is extended to $p\not=0$ in the following proposition.


\begin{proposition}\label{analytic}
For $\alpha$ such that $|z_+(\alpha)|<1$, and $p>0$, the function $G_{|\Psi_E\rangle}(z)$ defined by  $$G_{|\Psi_E\rangle}(z)=\frac{C_0y_1(\alpha)p}{y_2(\alpha)}(z-z_+)^B(z-z_-)^C\int_0^z{u^{p-1}(u-z_+)^{-(1+B)}(u-z_-)^{-(1+C)}du}$$ is analytic if and only if $B\in\mathbb{N}$ and $B\ge p $. This implies that the energies $E\in\sigma_\mathrm{p}(\mathcal{H}_\mathrm{ab}^{(\alpha)})$ must be discrete. In this case we have 
$$(1-2\alpha y)\sigma_\mathrm{p}(\mathcal{H}_\mathrm{ab}^{(\alpha)}) - \alpha y=\sigma_\mathrm{p}(\mathcal{H}_\mathrm{ab}).$$
\end{proposition}

\begin{proof} We will show in the appendix that there is a singularity of $G_{|\Psi_E\rangle}$ inside the unit disk which is removable under the condition $B\in\mathbb{N}$, $B\ge p$. Let us now show that $(1-2\alpha y)E - \alpha y$ must take values in the Bogoliubov spectrum \eqref{BogSpec2} for $B=m\in\mathbb{N}$. Suppose first that $p=0$. Then equation \eqref{eq:BandE} reads
\begin{equation}
\frac{(E+1)z_+(\alpha)+y_1(\alpha)}{z_+(\alpha)+2y_1(\alpha)} = -m.
\end{equation}
Using equation \eqref{eq:zeros} for $z_+(\alpha)$ and manipulating yields
\begin{equation}
E = \sqrt{1-4y_1(\alpha)y_2(\alpha)}(m+1) - \frac{2y_1(\alpha)y_2(\alpha)}{1-\sqrt{1-4y_1(\alpha)y_2(\alpha)}}.
\end{equation}
Recall the values of $y_1(\alpha) = \frac{y}{1-2\alpha y}$ and $y_2(\alpha) = \frac{y-\alpha+\alpha^2 y}{1-2\alpha y}$. Thus 
\begin{equation}
(1-2\alpha y)E - \alpha y = (m+1)\sqrt{1-4y^2}-\alpha y -\frac{2y(y-\alpha+\alpha^2 y)}{1-2\alpha y -\sqrt{1-4y^2}}.
\end{equation}
Interestingly, we can show that the quantity
\begin{equation}
Q := \alpha y + \frac{2y(y-\alpha+\alpha^2 y)}{1-2\alpha y -\sqrt{1-4y^2}}
\end{equation}
is independent of $\alpha$, which follows by factorizing the numerator. In fact, $Q = \frac{1}{2}\big(1+\sqrt{1-4y^2}\big)$. 
Hence, if $|\Psi_E\rangle$ is the eigenstate of $\mathcal{H}^{(\alpha)}_\mathrm{ab}$  with energy $E$ and generating function $G_{|\Psi_E\rangle}(z)$, then $e^{-\alpha a^\ast b^\ast}|\Psi_E\rangle$ is an eigenstate of Hamiltonian $\mathcal{H}^{(0)}_\mathrm{ab}$ with energy
\begin{equation}
(1-2\alpha y)E -\alpha y  = \frac{1}{2}\sqrt{1-4y^2}(2m+1)-\frac{1}{2}.
\end{equation}
This is exactly the result of Bogoliubov. The case $p\not= 0 $ is similar, so we skip many details. We now solve the equation
\begin{equation}
\frac{(p/2-E-1)z_+(\alpha)+y_1(\alpha)(p-1)}{z_+(\alpha)+2y_1(\alpha)} = m
\end{equation}
which gives
\begin{equation}
(1-2\alpha y)E - \alpha y = \sqrt{1-4y^2}(m-\frac{p}{2})+\frac{1}{2}\sqrt{1-4y^2} - \frac{1}{2}, 
\end{equation}
also in agreement with the Bogoliubov spectrum.\color{black}
\end{proof}

\subsubsection*{Bogoliubov spectrum for $|z_+|\ge1$:}
\label{subsec:escape}
When the generating function $G_{|\Psi_E\rangle}(z)$ which corresponds to eigenstate $|\Psi_E\rangle$ manifests no singularity in the unit disk, and thus the exponent $B$ has no constraint imposed by \eqref{eq:homogenous}, we instead show that $|\Psi_E\rangle\not\in\mathrm{dom}(e^{-\alpha a^\ast b^\ast})$ by analyzing the tranform of the generating function induced by the pair excitation operator.

\begin{proposition}\label{Transform}
Let $0<\alpha\le\alpha_c\le 1$ and $|\Psi_E\rangle\in\mathbb{F}_{ab}$ be a solution to the eigenvalue problem 
$$\mathcal{H}_\mathrm{ab}^{(\alpha)}|\Psi_E\rangle = E|\Psi_E\rangle.$$
Additionally, suppose that $|\Psi_E\rangle\in\mathrm{dom}(e^{-\alpha a^\ast b^\ast})$. If $G_{|\Psi_E\rangle}(z)$ is the generating function associated with $|\Psi_E\rangle$ in definition \eqref{genfunc}, then
$$G_{e^{-\alpha a^\ast b^\ast}|\Psi_E\rangle} (z) = \frac{1}{1+\alpha z}G_{|\Psi_E\rangle}\left(\frac{z}{1+\alpha z}\right).$$
\end{proposition}
It is important to note that the new generating function $G_{e^{-\alpha a^\ast b^\ast}|\Psi_E\rangle}(z)$ will have a region of analyticity that is in principle different from that of $G_{|\Psi_E\rangle}(z)$.
\begin{proof}
We introduce a some new notation that will aid in the proof. Given a sequence $\{a_s\}_{s=0}^\infty$, its \textit{ordinary generating function}, denoted $\mathrm{ogf}\{a_s\}(z)$ is a formal power series defined by 
$$\mathrm{ogf}\{a_s\}(z):= \sum_{s=0}^\infty{a_s z^s}.$$
Thus, $G_{|\Psi_E\rangle}(z)= \mathrm{ogf}\{C_s\}$ where $\{C_s\}_{s=0}^\infty$ are the rescaled coefficients of the expansion of $|\Psi_E\rangle$ in the particle number basis (definition \eqref{genfunc}).

The \textit{exponential generating function} for the same sequence is defined as
$$\mathrm{egf}\left\{a_s\right\}(z):=\sum_{s=0}^\infty{a_s\frac{z^s}{s!}}.$$
The following facts hold for generating functions: \textbf{(i)} The product of exponential generating functions can be written
$$\mathrm{egf}\left\{c_m\right\}(z)\cdot\mathrm{egf}\left\{d_m\right\}(z)=\mathrm{egf}\left\{\sum_{s=0}^m{c_m d_{m-s}\binom{m}{s}}\right\}(z).$$
The proof of this is direct. \textbf{(ii)} we can convert between ordinary and exponential generating functions by means of the \textit{Borel transform}, which is defined for an analytic function $f(z)$ by
$$\mathcal{L}'[f](z):=\int_{0}^\infty{f(zt)e^{-t}dt}.$$
The relevant fact for generating functions is
$$\mathcal{L}'\big[\mathrm{egf}\{d_m\}\big](z) = \mathrm{ogf}\{d_m\}(z).$$ Proof of this property is also direct and follows by noting that for $f(z)=\sum_{m=0}^\infty{f_m z^m}$, the action of $\mathcal{L}'$ on $f$ results in the formal multiplication of coefficient $f_m$ by $m!.$ 

We return to the generating function $G_{|\Psi_E\rangle}(z) = \mathrm{ogf}\{C_s\}$. Consider first the case $p=0$, so that $|\Psi_E\rangle$ lies in $\mathrm{span}\{|s,s\rangle\}_{s=0}^\infty$, and the rescaled coefficients of definition \eqref{genfunc} are the coefficients of $|\Psi_E\rangle$ in the particle number basis, $C_s=c_s$. Then by the assumption that $|\Psi_E\rangle \in\mathrm{dom}(e^{-\alpha a^\ast b^\ast})$, 
$$e^{-\alpha a^*b^*}|\Psi_E\rangle=\sum_{m=0}^\infty{\left(\sum_{s=0}^m{c_s \alpha^{m-s}\binom{m}{s}}\right)|m,m\rangle},$$
and so 
\begin{equation}\label{eq:ogf}
G_{e^{-\alpha a^\ast b^\ast}|\Psi_E\rangle}(z) = \mathrm{ogf}\left\{\sum_{s=0}^m{c_s(-\alpha)^{m-s}\binom{m}{s}}\right\}(z).
\end{equation}
Fact \textbf{(i)} shows that 
$$\mathrm{egf}\left\{\sum_{s=0}^m{c_s(-\alpha)^{m-s}\binom{m}{s}}\right\}(z)= \mathrm{egf}\left\{c_m\right\}(z)\cdot\mathrm{egf}\left\{(-\alpha)^m\right\}(z)=\mathrm{egf}\left\{c_m\right\}(z)\cdot e^{-\alpha z}.$$ 
while by fact \textbf{(ii)},
\begin{align*}
G_{e^{-\alpha a^\ast b^\ast }|\Psi_E\rangle}(z)&=\mathrm{ogf}\left\{\sum_{s=0}^m{c_s(-\alpha)^{m-s}\binom{m}{s}}\right\} = \mathcal{L}'[e^{-\alpha z}\mathrm{egf}\left\{c_m\right\}](z)\\
&= \int_{t=0}^\infty{\mathrm{egf}\{c_m\}(tz)e^{-t(\alpha z+1)}dt} \\
&= \frac{1}{1+\alpha z}\int_{\eta=0}^\infty{\mathrm{egf}\{c_m\}\left(\frac{z\eta}{\alpha z+1}\right)e^{-(\eta)}d\eta} \quad(\mathrm{for}\,\,\eta:= t(\alpha z+1))\\
&= \frac{1}{1+\alpha z}\mathcal{L}'[\mathrm{egf}\{c_m\}]\left(\frac{z}{1+\alpha z}\right) \\
&= \frac{1}{1+\alpha z}\mathrm{ogf}\{c_m\}\left(\frac{z}{1+\alpha z}\right) \\
&= \frac{1}{1+\alpha z}G_\Psi\left(\frac{z}{1+\alpha z}\right).
\end{align*}
This shows the proof for $p=0$. When $p\not=0$, the transformed state vector reads 
\begin{equation}
e^{-\alpha a^*b^*}|\Psi_E\rangle=\sum_{m=0}^\infty{\sqrt{\frac{(p+m)!}{m!}}\sum_{s=0}^m{c_s(-\alpha)^{m-s}\binom{m}{s}}|m+p,m\rangle}.
\end{equation}
It is straightforward to verify that the above computations go through using the rescaled coefficients $\{C_m\}_{m=0}^\infty$ and the generating functions
\begin{equation}
\begin{split}
G_{|\Psi_E\rangle}(z) &= \sum_{m=0}^\infty{C_m z^m}\quad\mathrm{and}\\
G_{e^{-\alpha a^\ast b^\ast}|\Psi_E\rangle}(z) &=\sum_{m=0}^\infty{\left(\sum_{s=0}^m{C_s(-\alpha)^{m-s}\binom{m}{s}}\right)z^m}.
\end{split}
\end{equation}
This concludes the proof.\color{black}
\end{proof}

We finally use Proposition \eqref{Transform} to recover the Bogoliubov spectrum in the case $|z_+(\alpha)|>1$.
Indeed, if the state $e^{-\alpha a^*b^*}|\Psi_E\rangle$ exists, it must have generating function $\frac{1}{1+\alpha z}G_{|\Psi_E\rangle}\left(\frac{z}{1+\alpha z}\right)$. The singularity of $G_{|\Psi_E\rangle}(z)$ at $z=z_+(\alpha)$ is therefore transformed to a singularity at $\frac{z}{1+\alpha z} = z_+(\alpha)$ of $G_{e^{-\alpha a^\ast b^\ast}|\Psi_E\rangle}(z)$, i.e.,
$$z=\frac{z_+(\alpha)}{\big(1-\alpha z_+(\alpha)\big)}.$$ 
Therefore, if $\Big|\frac{z_+(\alpha)}{1-\alpha z_+(\alpha)}\Big|<1$ it will be necessary that $B\in\mathbb{N}$ by similar reasoning as before. In addition to this singularity, $G_{e^{-\alpha a^\ast b^\ast}|\Psi_E\rangle}(z)$ has a removable singularity at $z = -1/\alpha$.

Since $z_+(\alpha)<0$ and $\alpha>0$ this is the same as showing that $\frac{1}{|z_+(\alpha)|}+\alpha >1$ or $|z_+(\alpha)|<\frac{1}{1-\alpha}$.  Indeed, recall
\begin{equation}
z_+(\alpha)=\frac{1}{2}\frac{(-1+2\alpha y+\sqrt{1-4y^2})}{(y-\alpha+\alpha^2y)}
\end{equation}
and 
$$y-\alpha+\alpha^2y=y\left(\alpha-\frac{1-\sqrt{1-4y^2}}{2y}\right)\left(\alpha-\frac{1+\sqrt{1-4y^2}}{2y}\right),$$ 
so that 
\begin{equation}
z_+(\alpha)=\left(\alpha-\frac{1+\sqrt{1-4y^2}}{2y}\right)^{-1}.
\end{equation}
For $0\le y\le\frac{1}{2}$ we have $\frac{1+\sqrt{1-4y^2}}{2y}\ge1$, hence 
\begin{equation}
|z_+(\alpha)|=\left(\left(\frac{1+\sqrt{1-4y^2}}{2y}\right)-\alpha\right)^{-1}\le\frac{1}{1-\alpha}.
\end{equation}
The zero $z_+(\alpha)$ therefore lies inside the unit disk.

Recall the relation for $E$ in terms of the exponent $B$
\begin{equation}\label{expB}
E+1-\frac{p}{2} = B+ \frac{y_1(\alpha)}{z_+(\alpha)}(2B+p-1).
\end{equation}
Using the expression for $z_+(\alpha)$ just derived and $y_1(\alpha) = \frac{y}{1-2\alpha y}$, we get
\begin{equation}
\frac{y_1(\alpha)}{z_+(\alpha)} = \frac{1}{2}\big(-1+\sqrt{1-4y^2}\big),
\end{equation}
and so 
\begin{equation}\label{E}
E = -\frac{1}{2}+\sqrt{1+4y^2}\Big(B+\frac{p-1}{2}\Big).
\end{equation}
In the case that $\alpha = \alpha_c$, we have $y_1(\alpha_c)=0$ and so by \eqref{expB}
$$E = B -1 +\frac{p}{2},$$
Thus the restriction $B\in\mathbb{N}$ yields half integer energies for $\mathcal{H}^{(\alpha = \alpha_c)}_\mathrm{ab}$. This is unique to choosing the value $\alpha = \alpha_c$; the spectrum of $\mathcal{H}^{(\alpha_c)}_\mathrm{ab}$ corresponding to transforms of LHY states is the same as the spectrum of $\frac{1}{2}(a^\ast a+ b^\ast b)$. Otherwise, the energy is given by \eqref{E}.

Note that the singularity due to $z_-(\alpha)$ remains outside the disk under the transformation $|\Psi_E\rangle \mapsto e^{-\alpha a^\ast b^\ast}|\Psi_E\rangle$. Indeed, 
\begin{equation}
z_-(\alpha) =\left(\alpha-\frac{1-\sqrt{1-4y^2}}{2y}\right)^{-1},
\end{equation}
and $\alpha\le\frac{1-\sqrt{1-4y^2}}{2y}\le 1$, so $\frac{1}{|z_-(\alpha)|}=|\alpha-\frac{1-\sqrt{1-4y^2}}{2y}|<1-\alpha$. This shows that the singularity due to $z_-(\alpha)$ does not affect the analyticity of $G_{e^{-\alpha a^\ast b^\ast}|\Psi_E\rangle}(z)$ inside the unit disk. 

\section{Completeness of the states $\{e^{(-\alpha_c a^\ast b^\ast)}|\Psi_{p,N}\rangle\}_{p,N\in\mathbb{N}}$}\label{completeness}
We now prove the completeness of eigenstates $\{e^{(-\alpha_c) a^\ast b^\ast}|\Psi_{p,N}\rangle\}_{p,N\in\mathbb{N}}$ in the Fock space $\mathbb{F}_{ab}$ for critical parameter $\alpha_c$. 
The completeness of these states is not obvious, since \textbf{(i)} The operator $e^{(-\alpha_c) a^\ast b^\ast}$ is unbounded, so the relation $$\mathcal{H}^{(\alpha_c)}_\mathrm{ab} = \frac{1}{1-2(\alpha_c) y}\Big(e^{\alpha_c a^\ast b^\ast}\mathcal{H}_\mathrm{ab}e^{(-\alpha_c) a^\ast b^\ast}+\alpha_c y\Big)$$
does not imply that the spectra of the two operators are the same, \textbf{(ii)} It is not clear that our formula \eqref{formula1} recovers the same degeneracy of eigenstates for $\mathcal{H}_\mathrm{ab}$ given in terms of tensor products of quasiparticle operators (a la equation \eqref{quasiparticle}), and \textbf{(iii)} we have determined that infinitely many eigenstates $|\Psi_{p,\Theta}\rangle$ corresponding to $\frac{p}{2}+\Theta\in\sigma(\mathcal{H}_\mathrm{ab}^{\alpha_c})$ and $\tilde{y}\ge 1$ cannot be in the domain of $e^{-\alpha_c a^\ast b^\ast}$.

The completeness result of this section implies that eigenstates of $\mathcal{H}_\mathrm{ab}$ are \textit{exactly} the non-Hermitian transforms of states with $\Theta\in\mathbb{N}$. The following density argument also gives an alternative proof to the claim that the degeneracy of the $E=\frac{p}{2}+N$ eigenspace matches the degeneracy of the momentum eigenstates $|p+N,N\rangle$, or $|N,p+N\rangle$.   



In order to show the completeness of the states $\{e^{-\alpha_c a^\ast b^\ast}|\Psi_{p,N}^{(\pm)}\rangle\}_{p,N\in\mathbb{N}}$ in $\mathbb{F}_{ab}$, it suffices to consider the completeness of states $\{e^{-\alpha_c a^\ast b^\ast}|\Psi_{p,N}\rangle\}_{N=0}^\infty$ in $\mathrm{span}\{|s+p,s\rangle,s\in\mathbb{N}\}$ for fixed $p$.
We make use of the well-known fact that a collection of elements $|\psi_j\rangle$ in a separable Hilbert space is complete if the only element $|f\rangle$ which annihilates every element of $|\psi_j\rangle$, that is $$\langle f |\psi_j\rangle = 0,\quad \forall j$$ is the zero element,  $|f\rangle \equiv 0$ \color{black}\cite{Lyubich}\color{black}.


We therefore fix $p\ge0$ without loss of generality. For $N\ge0$, a straightforward calculation gives: 
\begin{equation}\label{lkj}
e^{(-\alpha_c a^\ast b^\ast)} |\Psi_{p,N}\rangle = \sum_{m=0}^\infty {(-\alpha_c)^m S(m,N)|p+m,m\rangle}.
\end{equation}
for 
$$S(m,N) := \sum_{s=0}^{\min\{m,N\}}{(-\alpha_c \tilde{y})^{-s}\binom{N}{s}\sqrt{\binom{p+s}{s}}^{\,\,-1}\sqrt{\binom{p+m}{p+s}\binom{m}{s}}}.$$
If $|f\rangle$ is an arbitrary vector in $\mathbb{F}_\mathrm{ab}$ with expansion $|f\rangle:=\sum_{m}{d_m |p+m,m\rangle}$ and coefficients $\{d_m\}_{m=0}^\infty\in\ell^2$, the inner product of $|f\rangle$ with the state \eqref{lkj} reads:
\begin{equation}\label{eq:FinnerProd}
\langle f|e^{-\alpha_c a^*b^*} \Psi_{p,N}\rangle=\sum_{m}{\overline{d_m}(-\alpha_c)^m S(m,N)}.
\end{equation}

\begin{theorem}\label{Density}(Density of states for critical parameter)
Let $\alpha =\alpha_c =  \frac{1-\sqrt{1-4y^2}}{2y}$. 
Then the collection 
$$\{e^{-\alpha_c a^*b^*} |\Psi^{(\pm)}_{p,N}\rangle\}_{p\in\mathbb{Z},N\in\mathbb{N}}$$
is complete in the two-particle bosonic Fock space $\mathbb{F}_{ab}$.  
\end{theorem}

\begin{proof}
As described in the remark, it suffices to prove the statement for fixed $p$, and $|f\rangle$ having an expansion of the form above. Using equation \eqref{eq:FinnerProd} to write the inner product $\langle f|e^{-\alpha_c a^*b^*} \Psi_{p,N}\rangle$, the claim is that the infinite linear system for $\{d_m\}_{m=0}^\infty\in\ell^2$ which corresponds to the problem $\langle f|e^{(-\alpha_c a^\ast b^\ast)} \Psi_{p,N}\rangle=0\quad\forall N$ has only the trivial solution $\{d_m = 0\}_{m=0}^\infty\in\ell^2$. 

Summarizing the steps of the proof: \textbf{(i)} The family of functions defined by $f_N(z):=\langle f|e^{z a^*b^*}\Psi_{p,N}\rangle$ are shown to be analytic for $z$ in the unit disk. The main assumption, namely, $\langle f| e^{(-\alpha_c) a^\ast b^\ast}\Psi_{p,N}\rangle=0$ for all $N$, translates to the collection of point evaluatons $f_N(-\alpha_c)=0$ $\forall N$. \textbf{(ii)}  The derivatives $\frac{d^n}{dz^n}f_N(z)$ are written as linear combinations of $f_{N+j}(z)$ for $0\le j\le N+n$. Therefore the restriction $f_N(-\alpha_c)=0$ for all $N$ translates to the restriction on all derivatives, $\frac{d^n}{dz^n}f_N(-\alpha_c) = 0$ for $n,N\ge0$. \textbf{(iii)} Since $f_N(z)$ is analytic in the unit disk and $|\alpha_c|<1$ we conclude that $f_N(z)\equiv0$ for all $N$, which gives the proof.

Assume first $\tilde{y}=1$ (this condition will be removed later). Define the family of functions $f_N:\mathbb{C}\to\mathbb{C},N\in\mathbb{N}$ by making the substitution $-\alpha\mapsto z$ in \eqref{eq:FinnerProd}, i.e., 
\begin{equation}\label{eq:Nfamily}
\begin{split}
 f_N(z) &:= \sum_{m=0}^\infty{\overline{d_m}z^m\Bigg[\sum_{s=0}^{\min\{m,N\}}{z^{-s}\binom{N}{s}\sqrt{\binom{p+s}{s}}^{-1}\sqrt{\binom{p+m}{p+s}\binom{m}{s}}}\Bigg]} \\
 & = \sum_{m=0}^\infty{\overline{d_m}z^m\sqrt{\binom{p+m}{m}}\Bigg[\sum_{s=0}^{\min\{m,N\}}{z^{-s}\binom{N}{s}\binom{m}{s}\binom{p+s}{s}^{-1}}\Bigg]} \\
 & = \sum_{m=0}^\infty{\overline{d_m}z^m\sqrt{\binom{p+m}{m}}F\left(-m,-N,p+1;\frac{1}{z}\right)}.
\end{split}
\end{equation}

Here $F(a,b,c;z)$ is Gauss's Hypergeometric function \color{black}\cite{Andrews}\color{black}~ which is defined (using the Pochhammer symbols) by the power series 
$$F(a,b,c;z)\equiv\sum_{m=0}^\infty{\frac{z^m}{m!}\frac{(a)_m(b)_m}{(c)_m}},\quad a,b,c\in\mathbb{R},\quad c\not=0,\quad|z|<1,$$ and equation \eqref{eq:Nfamily} is a consequence of the factorization 
$$ \binom{p+m}{p+s}=\binom{m}{s}\binom{p+s}{s}^{-1}\binom{p+m}{m}.$$ 
Note that $f_0(z)\equiv\sum_{m=0}^\infty{\overline{c_m}\sqrt{\binom{p+m}{m}}z^m} $ is consistent with this definition since $F(\cdot,0,\cdot,z)=1$. It is easy to verify that the series defining $f_N(z)$ is convergent (and therefore $f_N$ is analytic) for $|z|<1$.

Now assume that $\langle f|e^{-\alpha_c a^\ast b^\ast} \Psi_{p,N}\rangle=0\quad\forall N$, which translates to the point evaluation
\begin{equation}\label{eq:ptEval}
f_N\big(-\alpha_c\big)=0\quad\forall N\ge0.
\end{equation}
The remainder of the proof describes how this condition implies $f_N(z)\equiv0\quad\forall N$ and hence $d_m\equiv 0$ for all $m$. We make use of the following relations for the Hypergeometric function $F(a,b,c;z)$ (see appendix for direct proofs):

\textbf{Property (i)} \textbf{(contiguous relation for the hypergeometric function)} For $p,N\ge 0$ 
\begin{equation}
\begin{split}
(mz)F\big(-m+1,-N,p+1;z\big) &= -(p+1+2N)F\big(-m,-N,p+1;z\big) \\
&+ (p+1+N)F\big(-m,-N-1,p+1;z\big) \\
&+ NF\big(-m,-N+1,p+1;z\big) 
\end{split}
\end{equation}
with the exceptional case $N=0$:
 \begin{equation}\label{eq:exception}
(mz)F\big(-m+1,0,p+1,z\big)=-(p+1)F\big(-m,0,p+1,z\big)+(p+1)F\big(-m,-1,p+1,z\big).
\end{equation}

Note: For clarity, we will use the conventional shorthand for the contiguous hypergeometric functions when the parameters $a,b,c$ are clear from the context:
\begin{equation}
F=F(a,b,c;z),\quad F(a\pm)=F(a\pm1,b,c;z),
\end{equation}
with similar definitions holding for $F(b\pm),\,\,F(c\pm)$.

\textbf{Property (ii)}  For $a<0,z\in\mathbb{C}$ 
\begin{equation}
\frac{d}{dz}{z^aF\big(a,b,c,\frac{1}{z}}\big)=az^{a-1}F(a+1,b,c,\frac{1}{z}).
\end{equation}

Using formulas \textbf{(i)} and \textbf{(ii)} we now show that $\frac{d^n}{dz^n}{f_N(z)}$ is a finite linear combination of $f_0(z),f_1(z),\dots,f_{N+n}(z)$ for every $n$. Since we have $f_n(-\alpha_c)=0$ for all $n$, i.e., equation \eqref{eq:ptEval}, this shows that $\frac{d^n}{dz^n}f_N(-\alpha_c)=0$ for all $n$. 

Indeed, 
\begin{equation} \label{fNderivative}
\begin{split}
\frac{d}{dz}{f_N^{(p)}}(z) &= \sum_{m=0}^\infty{\sqrt{\binom{p+m}{m}}\overline{d_m}\frac{d}{dz}\Big[z^mF\left(-m,-N,p+1,\frac{1}{z}\right)\Big]} \\
&= \sum_{m=0}^\infty{\sqrt{\binom{p+m}{m}}\overline{d_m}\Big[mz^{m-1}F\left(-m+1,-N,p+1,\frac{1}{z}\right)\Big]} \\
&= \sum_{m=0}^\infty{\sqrt{\binom{p+m}{m}}\overline{d_m}z^m\Big[m\frac{1}{z}F\left(-m+1,-N,p+1,\frac{1}{z}\right)\Big]} \\ 
&= \sum_{m=0}^\infty{\sqrt{\binom{p+m}{m}}\overline{d_m}z^m\Big[-(p+1+2N)F+(p+1+N)F(-N-1)+NF(-N+1)\Big]} \\
&= -(p+1+2N)f_N^{(p)}(z)+(p+1+N)f_{N+1}^{(p)}(z)+Nf_{N-1}^{(p)}(z).
\end{split}
\end{equation} 
Iterating this same procedure, it is clear that we can write a corresponding formula for any number of derivatives of $f_N$. This formula is written as a vector equation for convenience:
$$\frac{d^n}{dz^n}f_N^{(p)}(z) = \mathbf{L}\big(f^{(p)}_0(z),\dots,f^{(p)}_N(z),f^{(p)}_{N+1}(z),\dots,f^{(p)}_{N+n}(z)\big)^T,$$
where $\mathbf{L}$ is a $1\times(N+n+1)$ dimensional vector. The precise entries of $\mathbf{L}$ are irrelevant to us, except for the fact that they generally involve the constants $p,N$ but not $\alpha_c$. We should remark on the derivative of $f^{(p)}_0(z)$. In this case, equation \eqref{eq:exception} is used and Formula \textbf{(ii)} becomes trivial, so $\frac{d^n}{dz^n}f^{(p)}_0(z)$ will be a linear combination of the functions $f_0^{(p)},f_1^{(p)},\dots,f_n^{(p)}$.

Collecting like powers of $z$ in $f_N(z)$ gives
\begin{equation}
f_N(z) 
=\sum_{q=0}^{\infty}{\Bigg[\sum_{m=q}^{q+N}{\overline{d_m}\binom{N}{m-q}\binom{m}{q}}\Bigg]z^q},\quad\mathrm{for}\quad q=m-s,
\end{equation}
so indeed, $f_N\equiv 0$ implies that $d_m=0$ for all $m=0,1,2,\dots$ . 

Finally, the restriction $\tilde{y}=1$ is now removed via a change of variables. We show this for $p=0$ for clarity -- the other cases follow analogously. The function $f_N(z)$ is now defined by
\begin{equation}
f_N(z):=\sum_{m=0}^\infty{\overline{d_m}z^mF\Big(-m,-N,1,\frac{1}{\tilde{y}z}\Big)},\quad \tilde{y}\not=0.
\end{equation} 
We now have 
\begin{equation}
\frac{d}{dz}f_N(z)=\tilde{y}\sum_{m=0}^\infty{\overline{d_m}z^m\Big[m\frac{1}{\tilde{y}z}F\big(-m+1,-N,1,\frac{1}{\tilde{y}z}\big)\Big]},
\end{equation}
and formula \textbf{(i)} again applies to give 
\begin{equation}
\frac{d}{dz}f_N(z)=\tilde{y}\Big[-(1+2N)f_N(z)+(p+1)f_{N+1}(z)+Nf_{N-1}(z)\Big],
\end{equation} 
Following the above reasoning we conclude that $d_m=0$ for $m=0,1,2,\dots$ and the proof of density is complete.
\end{proof}

\section{A related particle-conserving Hamiltonian} \label{PC Ham}
We now briefly describe the construction of eigenstates for the particle-conserving approximation \eqref{H particle conserving}, i.e., 
$$
\mathcal{H}\approx 4\pi a \rho N + \sum_{k\in\mathbb{Z}^3_L, k\not=0}{\big\{k^2 +8\pi a\rho \big\}a_k^\ast a_k} + \frac{4\pi a}{|B_L|}\sum_{k\in\mathbb{Z}^3_L, k\not=0}{\big\{(a_0^\ast)^2 a_k a_{-k} + a_k^\ast a_{-k}^\ast (a_0)^2\big\}}.
$$
Since Wu describes a Hamiltonian which reduces to this operator in the periodic setting \color{black}\cite{Wu61,Wu98}\color{black}, we will refer to this approximation as ~$\mathcal{H}_\mathrm{Wu}$. The eigenstates of $\mathcal{H}_\mathrm{Wu}$ in the Fock space $\mathbb{F}_N$ will exhibit a structure similar to the states derived in Section \eqref{sec:construction}. The transformation in this setting is  $\exp(\mathcal{W})\mathcal{H}_\mathrm{Wu}\exp(-\mathcal{W})$, where the operator $\mathcal{W}$ is defined via
\begin{equation}
\mathcal{W} := \frac{1}{N}\mathcal{P}a_0^2.
\end{equation}
On the $N-$particle fiber, $\mathbb{F}_N$, the operators $\mathcal{W}$ and $\exp(-\mathcal{W})$ are bounded; we can therefore assert the equivalence
$$\sigma(\mathcal{H}_\mathrm{Wu}) = \sigma\big(\exp(\mathcal{W})\mathcal{H}_\mathrm{Wu}\exp(-\mathcal{W})\big).$$
The conjugations of the momentum basis operators with $\exp(\mathcal{W})$ are now more complicated than the conjugations with $\exp(\mathcal{P})$ fo equation \eqref{conjugation0}:
\begin{equation}\label{conjugation}
\begin{split}
\exp(\mathcal{W}) & a_0 \exp(-\mathcal{W}) = a_0,\\
\exp(\mathcal{W}) & a^\ast_k \exp(-\mathcal{W}) = a_k^\ast, \\
\exp(\mathcal{W}) & a_k \exp(-\mathcal{W}) = a_k +\frac{\alpha(k)}{N} a^\ast_{-k} a_0^2,  \\
\exp(\mathcal{W}) & a^\ast_0 \exp(-\mathcal{W}) 
= a_0^\ast + 2\mathcal{P} a_0.
\end{split}
\end{equation}
\color{black}
In particular, the transformed approximate Hamiltonian $\exp(\mathcal{W})\mathcal{H}_\mathrm{Wu}\exp(-\mathcal{W})$ will contain cubic and quartic terms in momentum state creation/annihilation operators; these terms must be dropped on the basis of steps \textbf{(i)} and \textbf{(ii)} of the approximation scheme in Section \eqref{sec:LY-H}. The result of \eqref{conjugation} and dropping the terms just described is
\begin{equation}
\begin{split}
\exp(\mathcal{W})\mathcal{H}_\mathrm{Wu} \exp&(-\mathcal{W}) \approx 4\pi a\rho N + 4\pi a\rho \sum_{k\in\mathbb{Z}^3_L}{\alpha(k)} + \frac{4\pi a}{|B_L|}\sum_{k\in\mathbb{Z}^3_L}{a_k a_{-k}(a^\ast_0)^2}\\
&+ \sum_{k\in\mathbb{Z}^3_L}{\Big(k^2+8\pi a \rho + 8\pi a \rho \alpha(k)\Big)(a_k^\ast a_k)} \\
&+ \sum_{k\in\mathbb{Z}^3_L}{\Big(\big(k^2+8\pi a\rho\big)\alpha(k)+4\pi a\rho + 4\pi a \rho \big(\alpha(k)\big)^2\Big)a_k^\ast a_{-k}^\ast \frac{a_0^2}{N} }.\\
\end{split}\color{black}
\end{equation}
The last line of this approximation, which contains the terms proportional to $(a^\ast_k a^\ast_{-k})$, will again vanish provided that equation \eqref{alphak} holds for $\alpha(k)$, which yields
\begin{equation} \label{TransformedWu}
\begin{split}
\exp(\mathcal{W})&\mathcal{H}_\mathrm{Wu} \exp(-\mathcal{W}) \approx  4\pi a\rho N + 4\pi a\rho \sum_{k\in\mathbb{Z}^3_L}{\alpha(k)} \\
&+\sum_{k\in\mathbb{Z}^3_L}{\Big(k \sqrt{k^2+16\pi a\rho }\Big)\big(a_k^\ast a_k \big)} + \frac{4\pi a}{|B_L|}\sum_{k\in\mathbb{Z}^3_L}{a_k a_{-k}(a^\ast_0)^2}.
\end{split}
\end{equation}
As before, we consider eigenstates $|\Psi(k)\rangle\in\mathbb{F}_N$ which are linear combinations of tensor product states containing only the momenta $(k,-k)$, and which solve the equation
\begin{equation}\label{WuEigenvalue}
\exp(\mathcal{W})\mathcal{H}_\mathrm{Wu} \exp(-\mathcal{W})|\Psi(k)\rangle = E|\Psi(k)\rangle.
\end{equation}
For this purpose, states $|\Psi\rangle\in\mathbb{F}_N$ are most succinctly described as vectors : 
\begin{equation}
|\Psi\rangle = \Big(\Psi_0,\Psi_1, \dots,\Psi_n,\dots,\Psi_{N-1},\Psi_N\Big)^T;
\end{equation} 
the component $\Psi_n$ in this vector refers to the symmetric tensor product that contains $N-n$ particles in the condensate $e_0(x)$, and $n$ particles in a state orthogonal to the condensate, denoted $\Psi
^{\perp}_n\in \phi^\perp(\mathbb{R}^{3n})$, viz.,  $$\Psi_n:=\Psi^{\perp}_n\otimes_{s}(e_0)^{\otimes_s N-n}.$$ 
In $|\Psi(k)\rangle$, we assume that the component $\Psi_n^\perp$ is a linear combination of momentum states $e_k(x)$ and $e_{-k}(x)$ for every $n$. In this notation, the eigenvalue problem \eqref{WuEigenvalue} translates to an upper triangular matrix eigenvalue problem for the vectors $(\Psi_0,\dots,\Psi_N)^T$. This is because the terms in the transformed Hamiltonian \eqref{TransformedWu} that correspond to momentum $k\in\mathbb{Z}^+_L$ either: \textbf{(a)} transform a component $\Psi_n$ to  $\widetilde{\Psi}_n$ with the same number of particles in the condensate, or \textbf{(b)} transform $\Psi_n$ to a $\widetilde{\Psi}_{n-2}$, with 2 additional particles in the condensate. 

The eigenvalues $E$ of this upper-triangular matrix equation are equal to the diagonal elements of the matrix, which are
\begin{equation}
E = \big(k\sqrt{k^2+16\pi a\rho}\big)n,\quad n = 0,1,\dots,N-1, N.
\end{equation}
These are energies of the diagonal operator $k\sqrt{k^2+16\pi a\rho}\big(a_k^\ast a_k + a^\ast_{-k}a_{-k}\big)$ on $\mathbb{F}_N$. The off-diagonal terms of \eqref{TransformedWu} couple states of the form $\Psi^\perp_n\otimes_s (e_0)^{\otimes_s N-n}$ to states $\Psi^\perp_{n-2}\otimes_s (e_0)^{\otimes_s N-n+2}$ where $\Psi^\perp_{n-2}$ is the result of annihilating one particle with momentum $k$ and one particle with momentum $-k$ from $\Psi^\perp_n$. 
Using the integer $p$ as before to denote the difference in the particle number between states with momentum $k$ and those with momentum $-k$, we can write the (non-normalized) eigenstates as 
$$
|\Psi_{p,N}(k)\rangle := \sum_{s=0}^{(N-p)/2}{\tilde{y}(k)^{-s}\binom{N}{s}\sqrt{\binom{p+s}{s}\binom{2N}{2s}\frac{1}{2s!}}^{\,\,-1}|s+p,s\rangle_{\mathbb{F}_N}},
$$
for $$ |s+p,s\rangle_{\mathbb{F}_N}:=\Big((e_k)^{\otimes_s s+p}\otimes_s(e_{-k})^{\otimes_s s}\otimes_s(e_0)^{\otimes_s N-2s-p}\Big)$$ and
 $$ \tilde{y}(k):= \frac{8\pi a}{|B_L| k\sqrt{k^2+16\pi a\rho}} .$$
Note that for every state in $|\Psi_{p,N}(k)\rangle$ there is a degenerate state $|\Psi_{p,N}^{(-)}(k)\rangle$ involving a linear combination of the tensor products $(e_k)^{\otimes_s s}\otimes_s (e_{-k})^{\otimes_s s+p}\otimes_s (e_0)^{\otimes_s N-2s-p}$. Therefore, the degeneracy of states with energy $E$ is $2N+1$, exactly the same as the degenerate subspaces of $\mathcal{H}_\mathrm{LHY}$. The expression $\exp(\mathcal{W})|\Psi_{p.N}(k)\rangle$ then gives the eigenstates of $\mathcal{H}_\mathrm{Wu}$.
 \color{black}
 
\section{Discussion and Conclusion}
 We conclude by suggesting a few extensions of the framework introduced in this paper. There are two apparent ways to extend the work here to study the bose gas in a more general trapping potential. \textbf{(i):} We can consider keeping the interaction potential $\upsilon(x,y)$ instead of replacing it by the pseudopotential, and introduce a compact cut-off of $\hat{\upsilon}$ in momentum space. Carrying out the steps of the approximation scheme then results in a quadratic Hamiltonian of the form
\begin{equation}\label{interaction2}
\mathcal{H}_\mathrm{app1}:= \sum_{k}{(k^2+\frac{N}{2V}\hat{\upsilon}(k))a^\ast_k a_k}+\frac{N}{2V}\sum_{k\in\mathbb{Z}^+_L, \,p\in I(k)}{\hat{\upsilon}(k,p)(a^\ast_k a^\ast_{p} + a_k a_{p}}).
\end{equation}
Here, $I(k)$ is a finite set in $-\mathbb{Z}^{+}_L$ that contains $-k$ for every $k\in\mathbb{Z}^+_L$. 
The more complicated coupling present means that we cannot use the model Fock space $\mathbb{F}_\mathrm{ab}$ to describe eigenstates of $\mathcal{H}_\mathrm{app1}$. We nonetheless hypothesize that the essential results of this work (regarding the structure \eqref{essential eigenstate} of eigenstates for $\mathcal{H}_\mathrm{app1}$) hold when using a non-unitary transformation of the form
$$\exp(\mathcal{Q}),\quad \mathcal{Q} := \sum_{k\in\mathbb{Z}^+_L,\,p\in I(k)}{\{-\alpha(k,p)a_k^\ast a_p^\ast\}}.$$

\textbf{(ii):} In \color{black}\cite{Grillakis Margetis Sorokanich}\color{black}~, we used a framework similar to that of section \eqref{PC Ham} to write the eigenstates of a particle conserving Hamiltonian on $\mathbb{F}_N$, which describes a dilute gas in a generic trapping potential $V(x)$. In in this setting the condensate $\phi(x)$ satisfies a Hartree-type equation. If $a_{\perp,x}$ denotes the field operator on $\phi^\perp$, the approximate Hamiltonian reads:
\begin{equation}\label{eq:H-app}
\mathcal{H}_\mathrm{app2} := NE_H + h(a_\perp^\ast,a_\perp) +\frac{(a_{\overline{\phi}})^2}{2N} f_\phi(a_\perp^\ast,a_\perp^\ast) +\frac{(a^\ast_{\phi})^2}{2N}\overline{f_\phi}(a_\perp,a_\perp) 
\end{equation}
where, e.g., $h(a_\perp^\ast,a_\perp) = \iint{dxdy\{ h(x,y)a^\ast_{\perp,x}a_{\perp,y}\}}$, and the corresponding kernels are
$$h(x,y):= \big\{-\Delta+V(x)+N(v\ast|\phi|^2)(x)\big\}\delta(x,y) + N\phi(x)\upsilon(x,y)\overline{\phi(y)}-\mu~,$$ for $\mu>0$,  and
$$f_\phi(x,y):= N\phi(x) \upsilon(x,y) \phi(y).$$
We showed that there is a \textit{non-orthogonal} basis $\{u_j(x)\}_{j\in\mathbb{N}}$ that plays the role of the momentum states $\{e_k(x)\}_{k\in\mathbb{Z}_L}$ in the construction of eigenstates of $\mathcal{H}_\mathrm{app2}$. It remains an open question as to whether these states can be utilized, in the spirit of the work presented here, to write many body excitations of the \textit{quadratic} approximation to $\mathcal{H}_\mathrm{app2}$:
\begin{equation}\label{Happ2}
\mathcal{H}_{\mathrm{app}2}\approx N E_{\text{H}}+h(a_\perp^{\ast}, a_\perp)+\frac{1}{2}f_\phi(a_\perp^{\ast},a_\perp^{\ast})+\frac{1}{2} \overline{f_\phi}(a_\perp,a_\perp).
\end{equation}
We note the similarity of \eqref{Happ2} to the approximate quadratic Hamiltonian of Fetter \color{black}\cite{Fetter, Fetter2009}.\color{black}
\section{Appendix}
\subsection*{Derivation of formula \eqref{formula1}}
Let $E\in \mathbb{C}$ denote the eigenvalue of $\mathcal{H}^{(\alpha_c)}_\mathrm{ab}$ associated with eigenvector $|\Psi_E\rangle\in\mathbb{F}_{ab}$. Expanding $|\Psi_E\rangle$ in the occupation number basis, i.e., 
\begin{equation}\label{expansion1}
|\Psi_E\rangle=\sum_{\{m_a,m_b\in\mathbb{N}\}}{c_{m_a m_b}|m_a,m_b\rangle},
\end{equation}
yields the following relation between coefficients
\begin{equation} \label{eq:cScheme1}
\frac{1}{2}(m_a+m_b)c_{m_a,m_b}+\tilde y \sqrt{(m_a+1)(m_b+1)}c_{m_a+1,m_b+1} = Ec_{m_a,m_b}.
\end{equation}
Next define $E_{m_a,m_b}:=\frac{1}{2}(m_a+m_b)$, which is the energy of the state $|m_a,m_b\rangle$ as an eigenvector of the operator $\frac{1}{2}(a^*a+b^*b)$, so that
\begin{equation}\label{eq:iteration}
c_{m_a+1,m_b+1}=\frac{\big(E-E_{m_a,m_b}\big)}{\tilde y \sqrt{(m_a+1)(m_b+1)}}c_{m_a,m_b}.
\end{equation}
Repeating this relation $s$ times results in the formula
\begin{equation}\label{eq:cScheme2}
c_{m_a+s,m_b+s} = \tilde{y}^{-s}\frac{(E-E_{m_a+s-1,m_b+s-1})\cdot\dots\cdot(E-E_{m_a,m_b})}{\sqrt{(m_a+s)\cdot \dots\cdot(m_a+1)(m_b+s)\cdot\dots\cdot(m_b+1)}}c_{m_a,m_b}.
\end{equation}
We now fix $\vec{m}:=(m_a,m_b)$, and define the energy difference $\Theta := E-E_{\vec{m}}$ as well as the shorthand $\vec{m}+s:= (m_a+s,m_b+s)$. Equation \eqref{eq:cScheme2} can be rewritten using the generalized binomial coefficient $\binom{\Theta}{s}:= \frac{\Gamma(\Theta+1)}{\Gamma(s+1)\Gamma(\Theta-s+1)}$ via
\begin{equation}\label{CmplusS}
\begin{split}
c_{\vec{m}+s} &= \tilde{y}^{-s}\frac{(\Theta-s+1)\cdot\dots\cdot(\Theta-1)\Theta}{\sqrt{(m_a+1)\cdot\dots\cdot(m_a+s)(m_b+1)\cdot\dots\cdot(m_b+s)}}c_{\vec{m}} \\
&=\tilde{y}^{-s}\binom{\Theta}{s}\frac{s!}{\sqrt{(m_a+1)\cdot\dots\cdot(m_a+s)(m_b+1)\cdot\dots\cdot(m_b+s)}}c_{\vec{m}} \\
&= \tilde{y}^{-s}\binom{\Theta}{s}\sqrt{\binom{m_a+s}{s}\binom{m_b+s}{s}}^{\,\,-1}c_{\vec{m}} .
\end{split}
\end{equation}
With the quantities $\Theta, \vec{m}$ fixed, the single coefficient $c_{\vec{m}}$ uniquely determines all other coefficients $c_{\vec{m}\pm s}$. Imposing the normalization condition
\begin{equation}\label{normalization}
\sum_{s=-\min{(m_a,m_b)}}^\infty{|c_{m_a+s,m_b+s}|^2} = 1,
\end{equation}
then determines $c_{\vec m}$ uniquely. Without loss of generality, we take $\vec{m}=0$.

\subsubsection*{Difference scheme for eigenstates of $\mathcal{H}^{(\alpha)}_\mathrm{ab}$}

An iteration scheme can be written in the particle number basis for eigenstates of $\mathcal{H}^{(\alpha)}_\mathrm{ab}$ when $0\le\alpha\le\alpha_c$. We again assume that the state $|\Psi_E\rangle$ with energy $E\in\mathbb{C}$ has the expansion \eqref{expansion1}. It suffices to fix $p\in\mathbb{N}$ and consider only linear combinations of states $|m_a,m_b\rangle = |p+s,s\rangle$ for $s = 0,1,2,\dots$, so that we can write $c_s := c_{p+s,s}$ and the eigenvalue equation reads:
\begin{align*}
\frac{1}{2}\sum_{s=0}^\infty{c_s(p+2s)|p+s,s\rangle} &+ y_1(\alpha)\sum_{s=0}^\infty{c_s\sqrt{(p+s)s}|p+s-1,s-1\rangle} \\
&+ y_2(\alpha)\sum_{s=0}^\infty{c_s\sqrt{(p+s+1)(s+1)}|p+s+1,s+1\rangle} \\
&= E\sum_{s=0}^\infty{c_s|p+s,s\rangle}.
\end{align*}  
For $E\in\mathbb{C}$ we arrive at the difference scheme
\begin{equation} \label{diffScheme2}
\left(\frac{p}{2}+s\right)c_s+y_1(\alpha)c_{s+1}\sqrt{(p+s+1)(s+1)}+y_2(\alpha)c_{s-1}\sqrt{(p+s)s}=Ec_s.
\end{equation}
\subsection*{Properties of the generating function $G_{|\Psi_E\rangle}(z)$} 
The ordinary differential equation satisfied by $G_{|\Psi_E\rangle}(z)$ is the result of multiplying equation \eqref{diffScheme2} by $z^s$ and summing over $s$; the solution $G_{|\Psi_E\rangle}(z) = G_{hom}(\Psi_E,z) + I_{\Psi_E}(z)$ follows by standard arguments. We now give a summary of the proof that for $p>0$ and $|z_+|<1$, the function $G_{|\Psi_E\rangle}(z)$ defined in Proposition \eqref{analytic} is analytic in the disk. This follows by direct computation.

By change of variables we write
\begin{equation}
G_{|\Psi_E\rangle}(z)=\frac{C_0y_1(\alpha)p}{y_2(\alpha)}(z-z_+)^B(z-z_-)^C\int_0^1{t^{p-1}(tz-z_+)^{-(1+B)}(tz-z_-)^{-(1+C)}dt},
\end{equation}
and utilize the two formulas: 
\begin{equation}
\Bigg(1-t\Big(\frac{z}{z_-}\Big)\Bigg)^{B-p} = \sum_{n=0}^\infty{\frac{\Gamma(-B+p+n)}{\Gamma(-B+p)}\frac{1}{n!}t^n\Big(\frac{z}{z_-}\Big)^n},
\end{equation}
as well as
\begin{equation}
\int_{0}^1{t^{p+n-1}\Big(1-t\frac{z}{z_+}\Big)^{-(1+B)}dt} = \frac{1}{(p+n)}F\Big(B+1, p+n; 1+p+n; \frac{z}{z_+}\Big),
\end{equation}
where $F$ is the hypergeometric function. A straightforward but tedious substitution of these formulas into $G_{|\Psi_E\rangle}(z)$ yields
\begin{equation}\label{genfunctionappendix}
\begin{split}
G_{|\Psi_E\rangle}(z) &= C_0 p \Big(\frac{1}{B}\Big)\Big(1-\frac{z}{z_-}\Big)^{p-1-B}\frac{\Gamma(-B)}{\Gamma(-B+p)}\Big(\frac{z}{z_+}\Big)^{-p}\Big(1-\frac{z_+}{z_-}\Big)^{-p} \mathcal{G}(z),
\end{split}
\end{equation}
where $\mathcal{G}(z)$ is an analytic function for $\Big|\frac{z}{z_-}\Big|<1$ given by 
\begin{equation}
\mathcal{G}(z):= \sum_{n=0}^\infty{\frac{\Gamma(-B+p+n)}{\Gamma(-B+p)}\frac{1}{n!}\Big(\frac{z}{z_-}\Big)^n F\Big(p+n-B, 1; 1-B; 1-\frac{z}{z_+}\Big)}
\end{equation}
Equation \eqref{genfunctionappendix} makes it clear that $G_{|\Psi_E\rangle}$ manifests a singularity unless $B=m\in\mathbb{N}$ and $B\ge p$. This is the basis of our assertion in the proposition.
\subsection*{Properties of the hypergeometric functions}

Here we provide direct proofs of properties \textbf{(i)} and \textbf{(ii)} used in the proof of Theorem \eqref{Density}.

 \textbf{Property (i)} For $p=0,1,2,\dots$ , $$mzF(-m+1)=-(p+1+2N)F+(p+1+N)F(-N-1)+NF(-N+1).$$
\begin{proof}
Using the definition of the hypergeometric function, the left-hand-side of the above relation reads
\begin{align*}
mz\sum_{s=0}^\infty{\frac{z^s}{s!}\frac{(-m+1)_s(-N)_s}{(p+1)_s}} &= \sum_{s=0}^\infty{\frac{z^{s+1}}{(s+1)!}\frac{(-m)_{s+1}(-N)_{s+1}}{(p+1)_{s+1}}\left\{\frac{m(-m+1)_s}{(-m)_{s+1}}\cdot\frac{(-N)_s}{(-N)_{s+1}}\cdot\frac{(p+1)_{s+1}}{(p+1)_{s}}\right\}} \\
&= \sum_{s=0}^\infty{\frac{z^{s+1}}{(s+1)!}\frac{(-m)_{s+1}(-N)_{s+1}}{(p+1)_{s+1}}\left\{\frac{-(s+1)(p+s+1)}{(-N+s)}\right\}}.
\end{align*}
The right-hand-side of \textbf{(i)} meanwhile reads
\begin{align*}
\sum_{s=0}^\infty{\frac{z^s}{s!}\frac{(-m)_s(-N)_s}{(p+1)_s}\left\{-(p+1+2N)+(p+1+N)\frac{(-N-1)}{(-N+s-1)}+N\frac{(-N+s)}{-N}\right\}}.
\end{align*}
It must be shown that the $(s+1)^{th}$ term of the RHS agrees with the $s^{th}$ term of the LHS, and that the $s=0$ term on the RHS vanishes. Indeed
\begin{equation}
\frac{(-m)_0(-N)_0}{(p+1)_0}\left\{-(p+1+2N)+(p+1+N)+N\right\}=0.
\end{equation}
This shows the second statement, while the coefficient of $\frac{z^{s+1}}{(s+1)!}\frac{(-m)_{s+1}(-N)_{s+1}}{(p+1)_{s+1}}$ on the RHS is 
\begin{align*}
&-(p+1+2N)+(p+1+N)\frac{-N-1}{-N+s}+N\frac{(-N+s+1)}{-N} \\
&=\frac{-(p+N+s+2)(-N+s)-(p+1+N)(N+1)}{(-N+s)} \\
&= \frac{-(s+1)(p+s+1)}{(-N+s)}
\end{align*}
The case $N=0$ follows by comparing the LHS
\begin{equation}
mz\sum_{s=0}^\infty{\frac{(-m+1)_s(0)_s}{(p+1)_s}\frac{z^s}{s!}}=mz
\end{equation}
to the RHS
\begin{align*}
-(p+1)\cdot1+(p+1)\sum_{s}{\frac{(-m)_s(-1)_s}{(p+1)_s}\frac{z^s}{s!}}=-(p+1)+(p+1)\left\{1+\frac{(-m)(-1)}{(p+s)}z\right\}=mz
\end{align*}
\end{proof}
\textbf{Property (ii):} For $a<0$, 
\begin{equation}
\frac{d}{dz}z^aF\left(a,b,c,\frac{1}{z}\right)=az^{a-1}F\left(a+1,b,c,\frac{1}{z}\right).
\end{equation} 
\begin{proof}
The LHS reads 
\begin{equation}
z^{a-1}\sum_{s=0}^\infty{\frac{(a)_s(b)_s}{(c)_s}\frac{z^{-s}}{s!}\left\{a-s\right\}}
\end{equation}
while the RHS reads
\begin{equation}
az^{m-s}\sum_{s=0}^\infty{\frac{(a+s)_s(b)_s}{(c)_s}\frac{z^{-s}}{s!}}.
\end{equation}
Proving the formula then follows from the relation $(a-s)(a)_s=a(a+1)_s$. The exceptional case $b=0$ is handled by
\begin{equation}
\frac{d}{dz}\left\{z^mF(-a,0,c,\frac{1}{z})\right\}=\frac{d}{dz}z^m.
\end{equation}
\end{proof}

\subsection*{Acknowledgements}
The author (S.S.) is indebted to Professor Dionisios Margetis for crucial discussions on the non-Hermitian scheme, as well as several of the analyticity arguments, that feature in this work.




\end{document}